\documentclass[letterpaper,11pt,reqno]{amsart}
\pdfoutput=1
\usepackage{etoolbox}
\providebool{j2col}
\setbool{j2col}{false}
\usepackage[english]{babel}%
\usepackage[utf8]{inputenc}
\usepackage[noadjust]{cite}
\usepackage{etoolbox}
\usepackage{amsmath, amsthm}
\usepackage{amsfonts, amssymb}
\usepackage{bm}

\usepackage{nicefrac}
\usepackage{pgf}
\usepackage[T1]{fontenc}
\usepackage[tracking]{microtype}
\usepackage[ttscale=0.81]{libertine}
\usepackage[libertine]{newtxmath}
\renewcommand{\mathbf}[1]{\bm{#1}}
\usepackage{fullpage}

\usepackage{graphicx}

\usepackage{hyperref}
\hypersetup{
  colorlinks=true,
  linkcolor=blue,
  urlcolor=blue,
  linktoc=all,
  citecolor=blue,
  bookmarksnumbered=true
}
\usepackage{cleveref}

\newtheorem{theorem}{Theorem}[section]
\newtheorem*{theorem*}{Theorem}

\newtheorem{fact}[theorem]{Fact}
\newtheorem{lemma}[theorem]{Lemma}
\newtheorem{corollary}[theorem]{Corollary}

\theoremstyle{definition}

\newtheorem{remark}{Remark}[section]

\crefname{theorem}{Theorem}{Theorems}
\crefname{observation}{Observation}{Observations}
\crefname{claim}{Claim}{Claims}
\crefname{condition}{Condition}{Conditions}
\crefname{example}{Example}{Examples}
\crefname{fact}{Fact}{Facts}
\crefname{lemma}{Lemma}{Lemmas}
\crefname{corollary}{Corollary}{Corollaries}
\crefname{definition}{Definition}{Definitions}
\crefname{remark}{Remark}{Remarks}
\crefname{proposition}{Proposition}{Propositions}
\newcommand{\st}[0]{\ensuremath{\;\mathbf{|}\;}}

\newcommand{\abs}[1]{\ensuremath{\left|#1\right|}}
\newcommand{\norm}[2][]{\ensuremath{\left\Vert #2 \right\Vert_{#1}}}

\newcommand{\ceil}[1]{\ensuremath{\left\lceil#1\right\rceil}}

\renewcommand{\Pr}[2][]{\ensuremath{\mathbb{P}_{#1}\insq{#2}}}
\newcommand{\E}[2][]{\ensuremath{\mathbb{E}_{#1}\insq{#2}}}

\newcommand{\inb}[1]{\left\{#1\right\}}
\newcommand{\inp}[1]{\left(#1\right)}
\newcommand{\insq}[1]{\left[#1\right]}

\makeatletter
\newcommand*{\defeq}{\mathrel{\rlap{%
                     \raisebox{0.3ex}{$\m@th\cdot$}}%
                     \raisebox{-0.3ex}{$\m@th\cdot$}}%
                    =}
\newcommand*{\eqdef}{=
  \mathrel{\rlap{%
      \raisebox{0.3ex}{$\m@th\cdot$}}%
    \raisebox{-0.3ex}{$\m@th\cdot$}}%
}
\makeatother

\newcommand{\nfrac}[3][]{\nicefrac[#1]{#2}{#3}}

\newcommand{\R}[0]{\ensuremath{\mathbb{R}}}
\DeclareMathOperator{\polyint}{poly}
\newcommand{\poly}[1]{\ensuremath{\polyint\inp{#1}}}

\renewcommand{\vec}[1]{#1}

\usepackage{url}

\newcommand{\enc}[0]{\ensuremath{\mathcal{E}}}
\newcommand{\dec}[0]{\ensuremath{\mathcal{D}}}

\newcommand{\vblock}[2]{\ensuremath{\textup{\textbf{bl}}(#1, #2)}}
\newcommand{\cblock}[2]{\ensuremath{\textup{\textbf{cl}}(#1, #2)}}

\newcommand{\lclock}[3][]{\ensuremath{\overline{\textup{\textbf{ll}}}_{#1}(#2, #3)}}
\newcommand{\mblock}[2]{\ensuremath{\textup{\textbf{Bl}}(#1, #2)}}
\newcommand{\mcblock}[2]{\ensuremath{\textup{\textbf{Cl}}(#1, #2)}}
\newcommand{\ran}[1]{\ensuremath{\textup{\textbf{Range}}(#1)}}

\newcommand{\tf}[0]{\ensuremath{\operatorname{\textbf{\textup{lg}}}}}
\newcommand{\lot}[1]{\ensuremath{#1}-lower triangular}
\newcommand{\C}[0]{\ensuremath{\mathcal{C}}}

\newcommand{\nset}{\mathbb N}

\notbool{j2col}{
  \makeatletter
  \newenvironment{flexeq}
 {%
   \def\eqbreak{}%
   \expandafter\equation
 }
 {%
   \expandafter\endequation
 }
 \newenvironment{flexeq*}
 {%
   \def\eqbreak{}%
   \expandafter\displaymath
 }
 {%
   \expandafter\enddisplaymath
 }
 \makeatother
 \newcommand{\eqbreak}{}
}{
  \makeatletter
  \newenvironment{flexeq}
  {%
    \def\eqbreak{\\}%
    \expandafter\multline
  }
  {%
    \expandafter\endmultline
  }
  \newenvironment{flexeq*}
  {%
    \def\eqbreak{\\}%
    \expandafter\multline
  }
  {%
    \nonumber\expandafter\endmultline
  }
  \makeatother
  \newcommand{\eqbreak}{\\}
}

\begin{document}
\title{Online codes for analog signals}
\thispagestyle{empty}
\author{Leonard J. Schulman \and Piyush Srivastava}
\let\thefootnote\relax\footnotetext{Leonard J. Schulman, California
  Institute of Technology. Email: \texttt{schulman@caltech.edu}.}
\let\thefootnote\relax\footnotetext{Piyush Srivastava, Tata Institute of
  Fundamental Research.  Email: \texttt{piyush.srivastava@tifr.res.in}.}
\footnotetext{Work supported in part by: United States NSF grants 1319745
  and 1618795; a Ramanujan Fellowship for the second author from SERB,
  Indian Department of Science and Technology; and a residency for the first
  author at the Israel Institute for Advanced Studies, supported by a EURIAS
  Senior Fellowship co-funded by the Marie Skłodowska-Curie Actions under
  the 7th Framework Programme.}

\maketitle

\begin{abstract}
This paper revisits a classical scenario in communication theory: a
waveform sampled at regular intervals is to be encoded so as to
minimize distortion in its reconstruction, despite noise. This
transformation must be online (causal), to enable
real-time signaling; and should use no more power than the original
signal. The noise model we consider is 
an ``atomic norm'' convex relaxation of the standard (discrete alphabet) Hamming-weight-bounded 
model: namely, adversarial $\ell_1$-bounded.
In the ``block coding'' (noncausal) setting, such encoding is possible due to the
existence of large almost-Euclidean sections in $\ell_1$ spaces, a
notion first studied in the work of Dvoretzky in 1961. Our main result is that an analogous result is achievable even
causally. Equivalently, our work may be seen as a ``lower
triangular'' version of $\ell_1$ Dvoretzky theorems.  %
In terms of communication, the guarantees are expressed in terms of
certain time-weighted norms: the time-weighted $\ell_2$ norm imposed on
the decoder forces increasingly accurate reconstruction of the
distant past signal, while the time-weighted $\ell_1$ norm on the
noise ensures vanishing interference from distant past noise.
Encoding is linear (hence easy to implement in analog
hardware). Decoding is performed by an LP analogous to those used in
compressed sensing.

{\textsc{Keywords:}}
  \emph{Online coding, Dvoretzky theorems, Analog signals, Random matrices.}

\end{abstract}

\setcounter{tocdepth}{1}
\tableofcontents

\newpage
\setcounter{page}{1}

\section{Introduction}
\label{sec:introduction-1}

\subsection{The problem} 
We study a fundamental scenario of
communication theory. A source is generating a waveform which we
sample at regular intervals. We wish to encode the signal in real
time, and decode the noise-affected transmission in real time, all while minimizing distortion in the reconstruction.

We require that the power (the $\norm[2]{\cdot}$ norm) of the transmission $\enc(x)$ not exceed a constant factor of the power  of the signal $x$; 
for simplicity of operation, we also wish the encoding map $\enc$ to be linear and deterministic. We operate in a worst-case model, namely, the adversary has  advance knowledge of the signal and its encoding.
Furthermore, at a minimum, we wish to have the following kind of decoding guarantee: for any signal $x: \nset \to \R$, and any bounded-power adversary noise $y: \nset \to \R$ (i.e., $\norm[2]{y}<\infty$), the ``limiting decoding'' $\dec(\enc(x)+y)$ should equal $x$.

Actually, much of our effort will be devoted to stronger results quantifying the rate with which the decoder can eliminate noise. 
For this we must examine more closely the strength of the adversary. 
A conventional approach in analog communications would be to allow noise $y$ such that $\norm[2]{y}$ is small compared with $\norm[2]{x}$. This is indeed the standard framework in analog
communication in which one first source-codes the signal using vector
quantization, then channel-codes the now discrete signal using a
finite alphabet, which, in turn, is encoded with a waveform. We would like, however to allow noise of comparable power to the signal, 
$\|y\|_2 \leq O(\|x\|_2)$, and even beyond.

It is immediately apparent, however, is that this kind of $\norm[2]{\cdot}$-bounded adversary is too powerful for the problem we consider:
 the adversary can assign $y=-\enc(x)$ and simply zero-out the transmission.

However,  a power constraint is only one plausible assumption on the noise source. The goal of our work is to show that if instead of the power constraint on the adversary, we make a different but very familiar assumption, we can provide an entirely different approach to this communication problem.

In our setting, where the noise is generated by an adversary, an alternative modeling assumption for noise is one
that has proven
fruitful as a model for \emph{signals} in the compressed sensing literature: it is that the
signal is bounded in (a possibly weighted) $\norm[1]{\cdot}$ norm.
An $\norm[1]{\cdot}$ norm bound (for a signal of given power)
is a kind of sparsity assumption, and sparsity is a natural characteristic
of many signal sources, which is in large part why this approach has
succeeded in compressed sensing~\cite{donoho_compressed_2006}. It is
therefore natural to pose the problem of protecting our signal against
interference by sparse signals generated by an adversary.
Indeed, in the context of digital error-correcting codes, the most
basic and prevalent model has long been of noise limited in Hamming
norm, which is precisely a sparsity assumption.  Relaxations of such
combinatorial sparsity assumptions to convex norms such as
$\norm[1]{\cdot}$ are also used to make them amenable to convex
programming formulations~\cite{CRPW12}.

Methodologically, the approach of considering adversaries bounded in the same norm as the signal has a fundamental limitation: no deterministic coding method can recover the signal to accuracy better than the signal-to-noise ratio. On the other hand, focusing on an adversary bounded in a different norm than the signal (here $\norm[1]{\cdot}$ rather than $\norm[2]{\cdot}$) opens the possibility of achieving in the limit noise-free decoding. That, as well as convergence rates to this limit, is the contribution of this paper: 
power-limited, real-time communications against an $\norm[1]{\cdot}$-bounded adversary.%
\footnote{As noted this cannot be achieved against \textit{general} $\norm[2]{\cdot}$-bounded adversaries. However, if an $\norm[2]{\cdot}$-bounded adversary eventually stops inserting noise, i.e., if $y$ has compact support, our decoding will be successful---reconstruction error will tend to $0$---because in this case the two norms are comparable.}

\subsection{An easier problem: block coding}
Undoubtedly, as for any error-correction problem, the most basic
problem which one must consider here
is that of
\textit{block coding} a signal. That is, the incoming signal is a
vector $\vec{x} \in \R^T$. We transmit at rate $1/\rho$, that is, we
map $\vec{x}$ to $\enc(x) \in \R^{\rho T}$. Our first constraint on
the encoder is an \emph{energy constraint}. If the encoder could
amplify the signal by an arbitrarily large factor, then it could swamp
out any interference by an adversary who is bounded in power or any other norm. Since this is an
unrealistic (and uninteresting) model for the encoder, we
stipulate that the total power of the transmission should be
comparable to the total power of the original message itself. That is, we ask that
\begin{equation*}
  \label{eq:30}
  \norm[2]{\mathcal{E}(\vec{x})_{[\rho t]}} \leq
  \norm[2]{\vec{x}_{[t]}}, \text{ for all $t \in
    [T]$}.\tag{*}
\end{equation*}
(Here $\vec{*}_{[t]}$ denotes the prefix of a vector
consisting of its first $t$ co-ordinates.) The noise source
adds a vector $y \in \R^{T\rho}$ onto $\enc(x)$; the noise $y$ may
depend upon $\enc(x)$. The receiver then applies a decoding map
$\dec\inp{\enc\inp{\vec{x}} + \vec{y}}$. The question is then what can
be achieved in terms of simultaneously
\begin{itemize}
\item Maximizing communication rate (minimizing $\rho$),
\item Minimizing distortion relative to noise, i.e., minimizing the ratio
$\frac{\norm[2]{
  \mathcal{D}(\mathcal{E}(\vec{x}) + \vec{y})
  - \vec{x}
}}{\norm[1]{\vec{y}}}$
\end{itemize}
As we discuss in more detail below, the answer to this question,
although not posed in this language, was given long ago in the work of
Milman~\cite{milman_new_1971}, Ka\v{s}in~\cite{Kashin77}, and Figiel, Lindenstrauss and
Milman~\cite{figiel_dimension_1977}, pursuing the study initiated by
Dvoretzky~\cite{dvoretzky61:_some_banac} of Euclidean
sections in Banach spaces.  Further, the codes so
achieved are \emph{linear}: the encoding operation consists of
multiplying the source vector by an appropriate $\rho T \times T$
matrix $A$, and the distortion ratio achieved is $O(T^{-1/2})$ (see
discussion leading to \cref{block-possible} below):
\begin{equation}
  \label{block-possibleA}
  \norm[2]{
    \mathcal{D}(\mathcal{E}(\vec{x}) + \vec{y})
    - \vec{x}
  } \leq O(T^{-1/2})\norm[1]{\vec{y}}.
\end{equation}

In contrast, our object of study in this paper is the \emph{real-time}
or \emph{causal} encoding and decoding of a source
generated on the fly, as for instance an audio signal, or the signal
from a remote sensor, in a distributed control setting.  While the
guarantees achieved in the offline (block coding) setting do serve as
a guideline for framing what might be achievable in online coding, it
will be clear from later discussions that not everything achievable offline can
be achieved in the online setting.

We now proceed to formulate the appropriate requirements for the
online setting.  The \emph{encoder} $\enc$ is required to be such that
the transmissions $1,\ldots,\rho t$ can depend only on the prefix
$\vec{x}_{[t]}$ of the message that is available to the source at time
$t$: in other words, $\rho \geq 1$ symbols are sent for each symbol of
the message, in a way such that these $\rho$ symbols depend only on
the prefix $\vec{x}_{[t]}$ of the message available at time $t$.  In
particular this enforces that
\begin{equation*}
  \label{eq:31}
  \mathcal{E}(\vec{x})_{[\rho t]} = \mathcal{E}(\vec{x}_{[t]}), \text{
    for all $t \in [T]$}.\tag{**}
\end{equation*}
The \emph{decoder} is now a collection of maps from $\R^{\rho t}$ to
$\R^t$ for each $t \in [T]$; the output of the decoder at time $t$ is
$\dec\inp{\enc\inp{\vec{x}_{[t]}} + \vec{y}}$, where
$\vec{y} \in \R^{\rho t}$ is the unknown error introduced by the
adversary up to time $t$.

As in the offline setting, we would like our encoder to be linear. The
requirement in \cref{eq:31} then implies that the matrix $A$
implementing the encoder needs to be lower triangular in the
rate-adjusted sense that $A_{i,j}=0$ if $i<j\rho$. (This is what we shall mean by ``lower triangular'' from here on.) However, none of
the constructions arising from the work on Euclidean sections cited
above provide a lower triangular $A$. This is to be expected since
our decoding requirement in \cref{block-possibleA}
is itself unreasonable in the online setting: for example, an
adversary who is silent for a while and then inserts a brief burst of
noise can satisfy the $\norm[1]{\cdot}$ bound over the history of the
communication, yet obliterate the last $\rho$ transmissions, which are
the only ones to carry information about the most recent portion of
the signal.

The above objection guides us toward the right decoding requirement in
the online setting. The idea is that the inaccuracy in the decoding of
a prefix $(\vec{x}_1,\ldots\vec{x}_t),$ of the signal should decrease
as time elapses after $t$, provided that the noise (even if
adversarial) is subject to a possibly time-weighted $\norm[1]{\cdot}$-norm
bound. Our aim is that for
that portion of the signal that is in the remote past, our decoding guarantee is
analogous to what can be achieved in the block coding setting. We now develop this idea quantitatively.

\subsection{Two inadequate definitions}
We start with two extreme formalizations, each of which captures one
desirable feature; and then combine these.
The first desideratum is that for any fixed $i$, as time $t$
goes on, our decoding of $x_i$ at time $t$ become ever-more accurate
provided that the noise is below tolerable limits. (And in particular
if the adversary stops injecting noise into the system.) This is
analogous to the decoding guarantee given for discrete alphabets by
tree codes.

We can formulate such a guarantee using a time-weighted norm for the
decoding error.  For a vector $\vec{x} \in \R^T$, we define the
$\norm[\star]{\vec{\cdot}}$ ``decoding norm'', in which the error on inputs from the remote past
is given higher weight than that on recent inputs:
\begin{equation}
  \label{eq:28}
  \norm[\star]{\vec{x}} \defeq
  \norm[\star(T)]{\vec{x}} \defeq
  \sqrt{\frac{1}{T}\sum_{i=1}^T(T - i + 1)x_i^2},
\end{equation}
and we modify the block-code decoding requirement
(\cref{block-possibleA}) to the following:
\begin{flexeq}
\label{eq:29}
\norm[\star(t)]{
  \mathcal{D}(\mathcal{E}(\vec{x}_{[t]}) + \vec{y})
  - \vec{x}_{[t]}
} \leq \frac{\norm[1]{\vec{y}}}{t^{1/2-\delta}}, \quad\eqbreak
\text{for all $t \in [T]$ and any given fixed $\delta \in (0, 1/2)$}.
\end{flexeq}
The flaw in this definition is that once the adversary has ever
injected noise into the system, no decoding is ever possible of
signals in the recent past (i.e., of $x_{t-c}$ at time $t$ for small
$c$), even if say the adversary has ceased to inject any noise after a
fixed time $t_0$. That is, requirement~\eqref{eq:29} fails a second desideratum: that the effects of any noise burst should dissipate over time.

This leads us to the other extreme: a decoding guarantee in which noise from the
distant past is allowed to contribute only vanishingly to the decoding
error.  For this we define the time-weighted ``noise norm''
$\norm[\dagger]{\vec{\cdot}}$:
\begin{equation}
  \label{eq:53}
  \norm[\dagger]{\vec{y}} \defeq
  \norm[\dagger(\rho T)]{\vec{y}} \defeq
  \sum_{i=1}^{\rho T} \abs{y_i} \inp{\frac{\rho T - i + 1}{\rho T}}^{-1/2},
\end{equation}
and impose again the decoding requirement
\begin{flexeq}
  \label{eq:54}
  \norm[2]{
    \mathcal{D}(\mathcal{E}(\vec{x}_{[t]}) + \vec{y})
    - \vec{x}_{[t]}
  } \leq \frac{\norm[\dagger(t)]{\vec{y}}}{t^{1/2-\delta}}, \notbool{j2col}{\quad}{\\}
\text{for all $t \in [T]$ and any given fixed $\delta \in (0, 1/2)$}.
\end{flexeq}
The flaw in this second definition is that it does not provide
gradually-improving decoding of each fixed input character (which was
the motivation for the first definition). For any fixed level of
noise, we have no better decoding guarantee on $x_1$ than on $x_T$ at
time $T$. (In particular, a bounded noise burst at time $T$ is enough
to ruin the decoding of $x_1$.)

\subsection{The satisfactory definition and our main result}
We achieve both desiderata with
a definition which time-weights both the adversary's noise and the decoding
error. Formally, for any $\mu \in [0, 1]$, define
\begin{equation} \begin{aligned}
  \label{eq:55}
  \norm[\star_\mu]{\vec{x}} & \defeq \norm[\star_\mu(T)]{\vec{x}}
  \defeq
  \sqrt{\sum_{i=1}^T\inp{\frac{T - i + 1}{T}}^\mu x_i^2} \\
  \norm[\dagger_{\mu}]{\vec{y}} &\defeq \norm[\dagger_\mu(\rho T)]{\vec{y}}
  \defeq \sum_{i=1}^{\rho T} \inp{\frac{\rho T - i + 1}{\rho
      T}}^{-(1-\mu)/2}\abs{y_i}.
\end{aligned} \end{equation}

This subsumes the earlier cases we considered:
\cref{eq:28} is the case $\mu = 1$ while \cref{eq:53} is the case
$\mu = 0$.

Now, given any choice of $\mu \in (0,1]$ and $\delta \in (0, 1/2)$, we
demand the decoding guarantee (generalizing \cref{eq:29,eq:54}):
\begin{equation*}
  \label{eq:56}
  \norm[\star_\mu(t)]{
    \mathcal{D}(\mathcal{E}(\vec{x}_{[t]}) + \vec{y})
    - \vec{x}_{[t]}
  } \leq
  \frac{\norm[\dagger_\mu(\rho t)]{\vec{y}}}{t^{1/2-\delta}}, \notbool{j2col}{\quad}{}
  \text{for all $t \in [T]$}.
 \tag{***}
\end{equation*}

Note that for any $\mu$, the penalty imposed by the
$\norm[\star_\mu]{\cdot}$-norm for errors made in decoding entries far
away in the past (say at times $s < c t$) is the same (to within a
constant factor $c' = c'(c)$) as that imposed by the
$\norm[2]{\cdot}$-norm.  Similarly, the weight assigned by the
$\norm[\dagger]{\cdot}$ norm to the adversary's noise inserted at
times $s < c t$ is within a constant factor to its unweighted
$\norm[1]{\cdot}$ norm.  However, when we are decoding entries $x_s$
for $s$ close to $t$, for which we do not yet have much information,
these weighted norms allow us to make larger errors in decoding
without much penalty. For $0<\mu \leq 1$, the requirement
\eqref{eq:56} on the decoder guarantees that as time progresses, so
does our ability to attenuate the error introduced by the adversary.
Further, in \Cref{prop:lower-bound}, we show that our requirements
enforce that the scaling of the attenuation factor in \eqref{eq:56} cannot be
$O(t^{-1/2})$ and must be of the form $\omega(t^{-1/2})$. In this the online coding problem differs from the block coding or $\norm[1]{\cdot}$-Dvoretzky problem.

Our main result is that for any fixed $\mu \in (0,1]$ and any
$\delta \in \inp{0, \frac{1}{2}}$, there is a constant-rate, constant-power code achieving
requirement~\eqref{eq:56}.
Our code is linear, and decoding too is efficient: the decoder solves a linear program analogous
to those appearing in the compressed sensing literature.

\noindent \textbf{Notation}.  For a $\rho T \times T$ matrix $A$, we
denote by $A_t$ the $\rho t \times t$ matrix consisting of its top
$\rho t$ rows and leftmost $t$ columns.
\begin{theorem}[\textbf{Informal, see \Cref{thm:encoding-matrix} for a
  formal statement}]
  \label{thm:main-informal}
  For any $\mu \in (0, 1]$ and $\delta \in (0, \frac{1}{2})$, there
  exists a rate parameter $\rho > 0$ for which there exists an encoder
  $\enc$ and a decoder $\dec$ satisfying the energy, error
  attenuation, and causal constraints in
  eqs.~\eqref{eq:30}, \eqref{eq:31}, and \eqref{eq:56}.

  In particular, the encoder \enc\ acts as left multiplication by a
  $\rho T \times T$ matrix $\C$ that is rate-adjusted lower triangular (i.e.,
  $\C_{ij} = 0$ when $i \leq (j-1)\rho$), where $T$ is the total time
  of transmission.  At time $t \leq T$, the decoder $\dec$ acts by
  making an $\norm[\dagger]{\cdot}$-norm projection to $\ran{\C_t}$ and then
  applying $\C_t^{-1}$ (which is well defined on $\ran{\C}$).
\end{theorem}

A natural special case of the above is with $\mu=1/2$. In this case the decoder and encoder guarantee
\begin{flexeq*}
\sqrt{\sum_{i=1}^T\left(\frac{T - i + 1}{T}\right)^{1/2} \vert 
\mathcal{D}(\mathcal{E}(\vec{x}_{[t]}) + \vec{y})_i
-x_i\vert^2}
\eqbreak
\leq O(T^{\delta-1/2}) \sum_{i=1}^{\rho T} \left(\frac{\rho T - i +
    1}{\rho T}\right)^{-1/4}\vert{y_i}\vert
\end{flexeq*}

In particular, if both the waveform values $x_i$ and the noise values $y_i$ are
$\Theta(1)$, then the error incurred by the decoder on a given entry of the
signal, decreases to zero (at almost a $T^{-1/2}$ pace) as the communication
continues in time.  {We also note that
 the quantity $(1-\mu)$ appearing in the exponent of the gain factor
  used in the $\norm[\dagger_\mu]{\cdot}$ norm cannot be replaced by any
  strictly smaller quantity; see the remark following
  Theorem~\ref{thm:encoding-matrix} for details.}
\subsection{Block coding and the Dvoretzky theorem}
\label{sec:block-coding-dvor}

We now revisit the relation between Euclidean sections and block
coding briefly alluded to above.
Our goal in this paper may also be framed as showing the existence of
a ``lower triangular'' analogue of a Euclidean section. 
This lower triangular constraint is the main source of technical
difficulty in our work as compared to previous work; in particular,
our method is quite different. The prior work does, however, show
some limits on what can be achieved: specifically, it is enough to imply that
the parameter
$\delta$ in \Cref{thm:main-informal} has to be non-negative.
In recent years, the classic work on Euclidean sections has been
re-interpreted explicitly in coding-theoretic language in a line of
work that seeks to derandomize the original
constructions~\cite{artstein-avidan_logarithmic_2006,lovett_almost_2008,
  guruswami_euclidean_2008,guruswami_almost_2010,IndykS2010}.  We now sketch
these connections.

Dvoretzky~\cite{dvoretzky61:_some_banac} initiated the study of the
existence of large subspaces $S$ of $\R^n$ equipped with an arbitrary
norm which are ``close'' to being Euclidean. Our interest here is in
the case where the norm is an $\norm[p]{\cdot}$-norm with $p = 1$, in
which case the condition of $S$ being close to Euclidean can be
written as
\begin{displaymath}
  \sup_{\vec{x} \in S}\frac{\sqrt{n}\norm[2]{\vec{x}}}{\norm[1]{\vec{x}}}
  \leq \Delta.
\end{displaymath}
Here $\Delta$ is the distortion of the section, and one seeks to make
it as close to $1$ as possible.  The problem of finding Euclidean
sections of large dimension has also been extensively studied,
starting with the work of Figiel, Lindenstrauss and Milman, and of
Ka\v{s}in~\cite{milman_new_1971,Kashin77,figiel_dimension_1977}, and in
the special case $p = 1$ it is known that there exists a constant
$c > 1$ (depending on $\Delta$) such that $(\R^{cn}, \norm[1]{\cdot})$
contains an Euclidean section of dimension $n$.

An equivalent view of Euclidean sections can be obtained in terms of a
modified ``condition number'' of appropriate tall matrices~(see, e.g.,
\cite{figiel_dimension_1977}).  In particular, if there exists a real
$cn \times n$ matrix $A$ of rank $n$ such that
\begin{equation}
  \norm[2\rightarrow 2]{A}\cdot \norm[1\rightarrow 2]{A^{-1}}
  \leq \frac{\Delta}{\sqrt{n}} \label{eq:32}
\end{equation}
then $(\R^{cn}, \norm[p]{\cdot})$ has a Euclidean section of dimension
$n$ (namely, \ran{A}) with distortion at most $\Delta$ (Here, and
subsequently, $A^{-1}$ denotes Moore-Penrose
pseudo-inverse).  It is also not hard to see that the existence of
such a Euclidean section implies the existence of a rank $n$
$cn \times n$ matrix $A$ satisfying \cref{eq:32}.

This representation of an Euclidean section allows us to view it as a
``block'' version of the codes we seek in this paper.  For, let $A$ be a
matrix satisfying the constraint in \cref{eq:32}, and assume without
loss of generality that $\norm[2\rightarrow 2]{A} = 1$ (this can be
ensured since the requirement in \cref{eq:32} is invariant under
scaling $A$ by constants).  Define the encoder \enc\ as left
multiplication by $A$: $\enc(\vec{x}) = A\vec{x}$.  The decoder $\dec$
acts on an input $\vec{y}$ by first finding the point $\vec{y'}$ in
\ran{A}\ that is closest to $\vec{y}$ in the $\norm[1]{\cdot}$-norm
(choosing one arbitrarily if there are several such points), and then
returning $A^{-1}\vec{y'}$.  Since $\norm[2 \rightarrow 2]{A} = 1$,
the energy constraint (\cref{eq:30}) is satisfied automatically.
Using \cref{eq:32} it can also be shown that
\begin{equation}
  \label{block-possible}
  \norm[2]{
    \mathcal{D}(\mathcal{E}(\vec{x}) + \vec{y})
    - \vec{x}
  } \leq \frac{2\Delta}{\sqrt{n}}\norm[1]{\vec{y}}.
\end{equation}
It is this guarantee for block decoding that we compare our
result in \Cref{thm:main-informal} against.

\subsection{Related work}
\label{sec:related-work}
We are following here on two main lines of work in communications. One is the investigation begun by Sahai and Mitter of the
``anytime capacity'' of a communication channel, which they discovered
to be essential to the feasibility of using that channel to control an
unstable plant in real time~\cite{SahaiM06}. Several types of channels and noise have been
studied but the primary concern in that literature is the role of channel noise in a feedback loop,
and to our knowledge there is no result which resembles ours.
The second concerns real-time communication of discrete signals over
discrete channels; one of the results from that literature is that it
is possible to causally encode a signal in such a manner that at all
times $T$, if the noise has so far corrupted only $cT$ characters
($c>0$ sufficiently small), then the decoder can correctly determine
the initial $(1-O(c))T$ characters~\cite{Sch96}. Our main result in this
paper is intended as the appropriate analog of the latter statement for a physical signal
and a physical channel, where ``characters'' are amplitudes of
a  waveform.

Our proof of existence of the code proceeds through an analysis of
certain random matrices with independent but not identically
distributed Gaussian entries.  In this light, our requirements,
especially when rephrased in terms similar to \cref{eq:32}, are
connected to the long line of work on the condition number of almost
square (and even square) random matrices (see, e.g,
\cite{edelman_tails_2005,rudelson_invertibility_2008,rudelson_smallest_2009,tao_inverse_2009}
and the references therein).  Note, however, that we are concerned
here with an analogue of a $\norm[1 \rightarrow 2]{\cdot}$-norm of the
pseudo-inverse of the encoding matrix, while in the work on condition
number the emphasis is on the $\norm[2 \rightarrow 2]{\cdot}$ norm of
the inverse.  Further, much of the work on the condition number has
considered rectangular random matrices with identically distributed
entries (see, however, the work of Cook~\cite{cook_lower_2016} and
Rudelson and Zeitouni~\cite{rudelson_singular_2016} for recent
progress on the lowest singular value of a class of structured
matrices with non-i.i.d. entries) while we are in a very different
regime---the main technical challenge of our work is to deal with the
pseudo-inverse of random lower triangular matrices (whose non-zero
entries are also not identically distributed).  Nevertheless, we
believe that the techniques developed in the work on the condition
number may be relevant for further improvements of our result,
especially on the question of achieving an optimal rate.  We also note
in passing that if one is concerned only with the norm of a matrix
with independent but not necessarily i.i.d.\ entries (rather than the
norm of its pseudo-inverse) then there are results in the literature
providing good asymptotic bounds~(see, e.g.,
\cite{latala_estimates_2005,schutt_expectation_2013,bandeira_sharp_2014,latala_dimension-free_2017}).
Our analysis in fact uses one of these bounds from the work of
Bandeira and van Handel~\cite{bandeira_sharp_2014}.

A rather different notion of online coding
underlies the long and celebrated line of work on fountain
codes~\cite{byers_digital_1998}.
Recall that in our online coding setting, (1) the encoder does not
receive the message symbols as a block but in an online fashion; (2)
the adversary corrupts transmitted symbols rather than erasing them,
so that the receiver does not know if a received symbol is corrupted
or not; (3) both the original message and the transmission have
real numbers as symbols. In fountain codes and continuing work such as LT
codes~\cite{luby_lt_2002} and Raptor
codes~\cite{shokrollahi_raptor_2006} (see also
\cite{maymounkov_online_2002}), the setting is different:  (1) even at time $t=0$, the
encoder has access to the full block of $n$ symbols comprising the
message; (2) the message is to be sent over an erasure channel;
(3) both the source and transmission symbols come from a discrete
alphabet.  The goal is for the code to be online in the sense that the
encoder generates a potentially infinite number of symbols using a
randomized algorithm, in such a way that the generated symbols are
mutually independent random variables, but the receiver is able to
decode the message with high probability as soon as it gets access to
\emph{any} $\Theta(n)$ of the encoder's generated symbols.  Fountain
codes and refinements such as LT and Raptor codes allow for very fast
encoding and decoding while achieving the above goal.

\subsection{Discussion}
\label{sec:discussion}
The most fundamental open question left open by our work is no doubt that of an explicit construction. On the positive side,
 the random matrices used in our constructions have with positive constant
probability
the properties we require. However, a more explicit construction that reduces the
dependence on randomness, and more importantly enables efficient verification of the properties, is desirable. The ideas involved
in the partial derandomizations of Euclidean
sections~\cite{artstein-avidan_logarithmic_2006,lovett_almost_2008,guruswami_euclidean_2008,guruswami_almost_2010, IndykS2010} or in tree code constructions~\cite{EvansKS94,Braverman12,moore2014tree,CohenHS18}
may help toward this goal.

It is also likely that the tradeoff we provide between the rate $1/\rho$ of the code  and the
$\tilde{O}(t^\delta)$ overhead in Theorem~\ref{thm:main-informal} can be improved; such optimization will be important toward practical implementation.

A third and fascinating question is whether the LPs to be solved in each decoding round, can be solved more quickly (at least in an amortized sense) thanks to the ``warm start'' from the only-slightly-different LP solved in the previous round.

\section{Online codes and low distortion matrices}
\label{sec:online-codes-low}

In this section, we provide a more quantitative discussion of the
connection of our work to Euclidean sections of $\ell_1$.  We start
with setting up some preliminary notation, and then state our main
technical theorem (\Cref{thm:encoding-matrix}), which establishes the
existence of a lower triangular analogue of an Euclidean section.  We
then show that this implies the existence of the codes we seek. The
rest of the paper is then devoted to proving
\Cref{thm:encoding-matrix}.

\subsection{Notation}
\label{sec:rate-functions-lotk}
Given a positive integer $k$ a \emph{\lot{k} matrix} $M$ with $T$
columns is a $kT \times T$ matrix in which $M_{ij} = 0$ if
$i \leq (k-1)j$.  For convenience, we also index the rows of such a
matrix by ordered pairs $(i, l)$ where $i \in [T]$ and $l \in [k]$,
and the row indexed $(i, l)$ is the $((k-1)i + l)$-th row from the
top.  The \lot{k} condition can then be stipulated more succinctly as
\begin{equation}
  \label{eq:52}
  M_{(i,l), j} = 0 \text{ when } i > j.
\end{equation}

\subsection{The main theorem and the code}
\label{sec:main-theorem-code}
Note that left multiplication of a message vector $\vec{x}$ by a
\lot{k}\ matrix $M$ satisfies the ``online'' or ``causal'' constraint
referred to in the introduction.  The next theorem shows that
there exists such a \lot{k}\ matrix with properties which imply the
other properties asked of the code in the introduction.

\begin{theorem}[\textbf{The encoding matrix}]
  \label{thm:encoding-matrix} 
  For any $\mu \in (0,1]$ and $\delta \in (0, \frac{1}{2})$, there
  exist positive constants $c_0, k_0$ such that the following is true.
  Let $T \geq 3$ be any integer.  For any rate parameter $k \geq k_0$
  there exists a \lot{k} $kT \times T$ matrix $\mathcal{C}$ satisfying
  the following conditions. (Recall that we denote by $\mathcal{C}_t$
  the $kt \times t$ leading principal submatrix of $\mathcal{C}$.)
  \begin{enumerate}
  \item  \label{item:C-operator-norm}
    \textbf{\textup{Submatrices of $\C$ have small operator norm}}:
    $\norm[2 \rightarrow 2]{\C_t} \leq 1$ for $1 \leq t \leq T$.
  \item  \label{item:C-invert}
    \textbf{\textup{Submatrices of $\C$ are robustly invertible}}: %
    for $1 \leq t \leq T$,
    \begin{displaymath}
      \norm[\dagger_\mu(kt)]{\C_t\vec{x}} \geq {c_0 t^{(1/2 - \delta)}}
      \norm[\star_\mu(t)]{\vec{x}} \text{ for all $\vec{x} \in \R^t$}.
    \end{displaymath}
  \end{enumerate}
\end{theorem}
The proof of this theorem will be through an analysis of certain \lot{k}\ random
matrices with independent but \emph{not} identically distributed Gaussian
entries.  In the next section (\Cref{sec:preliminaries-and-overview}), we start
with a simplified overview of the proof, before proceeding with the complete
proof in \Cref{sec:progr-atten-gauss}.  Here, we will show how the theorem
immediately yields a code satisfying the conditions outlined in the
introduction.  {But, first, we make a couple of remarks on the
  choice of the norms $\norm[\star_\mu]{\cdot}$ and $\norm[\dagger_\mu]{\cdot}$,
  and on the comparison between the respective robust invertibility guarantees
  that can be made in the online and block coding settings.}
\begin{remark}
  {We argue, by considering the action of the code on a unit pulse $e_t$ at
    time $t$, $(1 \leq t \leq T)$, that the quantity $(1-\mu)$
    appearing in the exponent of the gain factor used in the
    $\norm[\dagger_\mu]{\cdot}$ norm cannot be replaced by any strictly smaller
    quantity independent of $\delta$. To see this, observe that when $x = e_t$,
    the right hand side of item~\ref{item:C-invert} of the theorem is
    $\Theta(t^{(1-\mu)/2 - \delta})$.  On the other hand, due to the online
    encoding requirement~\eqref{eq:30}, $\C_te_t$ must be a vector in $\R^{k t}$
    in which only the last $k$ entries may be non-zero.  Further, the power
    constraint requirement~\eqref{eq:31} implies that these non-zero entries are
    $O(1)$.  It follows that if the quantity $(1-\mu)$ in the definition of the
    $\norm[\dagger_\mu({k t})]{\cdot}$-norm is replaced by $\tau$, the left hand
    side of item~\ref{item:C-invert} is at most $O(t^{\tau/2})$.  Thus for the
    inequality in item~\ref{item:C-invert} to be possible for all $\delta > 0$,
    one requires that $\tau \geq (1-\mu)$.}
\end{remark}
\begin{remark}
  The robust invertibility guarantee obtained for the encoding
  matrices constructed in \Cref{thm:encoding-matrix} falls short of
  the guarantee obtainable in the block coding setting (\cref{eq:32}),
  in the sense that we lose an extra $\Theta(t^\delta)$ factor in the
  online setting, albeit with the option to choose $\delta > 0$ as
  close to zero as we please at the cost of a deterioration in the
  rate of the code. A natural question therefore is whether it is
  possible to get rid of this loss and obtain a guarantee as strong as
  the block coding setting in the online setting as well. In the
  following theorem (proved in \Cref{sec:lower-bound}), we show
  that it is \emph{not} possible to obtain the guarantee of
  (\cref{eq:32}) in the online coding setting, and a loss of an
  $\omega_t(1)$ factor in the robust invertibility criterion must be
  incurred if the power constraint is to be satisfied.
  \begin{theorem}
    \label{prop:lower-bound}
    Fix $\mu \in (0, 1]$, $c_0 > 0$ and a positive integer $k$.  There
    exists a constant $\tau = \tau(\mu, c_0, k)$ such that the following
    is true.  If $\C$ is a $kT \times T$ \lot{k} matrix such that for
    all $t \in [T]$ the submatrix $\C_t$ of $\C$ satisfies
    \begin{equation*}
      \norm[\dagger_\mu(kt)]{\C_t\vec{x}} \geq {c_0 t^{1/2}}
      \norm[\star_{\mu}(t)]{\vec{x}} \text{ for all $\vec{x} \in \R^t$},
    \end{equation*}
    then there exists a non-zero $\vec{x} \in \R^T$ for which
    \begin{displaymath}
      \norm[2]{\C\vec{x}}^2 \geq \tau\sum_{i=1}^T\frac{1}{i}
      \geq \inp{\tau\log T}\cdot\norm[2]{\vec{x}}^2.
    \end{displaymath}
    \end{theorem}

(An open question left by our work is to narrow the gap between our upper bound of $O(t^\delta)$ and our lower bound of $\Omega(\sqrt{\log t})$, on the norm loss due to the causal-coding restriction.)
\end{remark}

We now show how \Cref{thm:encoding-matrix} immediately yields a code
satisfying the conditions outlined in the introduction.  Let $T$ be
the total time of transmission, and for $\delta \in (0, 1/2)$ let
$\mathcal{C}$ be a \lot{k}\ matrix with the rate parameter $k$ as in
the theorem.  The encoder $\mathcal{E}$ is defined as left
multiplication by the $k t \times t$ leading principal submatrix of
$\mathcal{C}$:
\begin{displaymath}
  \mathcal{E}(\vec{x}) = \C_t\vec{x} \text{ for all $\vec{x} \in \R^t, 1\leq t
    \leq T$.}
\end{displaymath}
Thus, the encoder only needs to send $k$ symbols at each time $t$.

At time $t \leq T$, the decoder $\mathcal{D}$ acts as follows.  Given
a received message $\vec{z} \in \R^{kt}$, it outputs the solution
$\vec{x_0} \in \R^t$ of the following linear program:
\begin{equation}
  \label{eq:43}
  \norm[\dagger_\mu(k t)]{\vec{z} - \C_t\vec{x_0}} \leq \min_{\vec{z'} \in \ran{\C_t}}
  \norm[\dagger_\mu(k t)]{\vec{z} - \vec{z'}}.
\end{equation}

We now show that the code $\C$ satisfies the conditions (\ref{eq:30})-(\ref{eq:56}). The online encoding condition, 
\cref{eq:31}, holds by construction since $\C$ and its
submatrices $\C_t$ are
\lot{k}.  The power constraint, \cref{eq:30}, is satisfied since for each $1
\leq t \leq T$ and any $\vec{x} \in \R^t$, applying Theorem~\ref{thm:encoding-matrix}(\ref{item:C-operator-norm}),
\begin{displaymath}
  \norm[2]{\mathcal{E}\vec{x}}
  \leq \norm[2 \rightarrow 2]{\C_t} \cdot \norm[2]{\vec{x}}
  \leq \norm[2]{\vec{x}}.
\end{displaymath}

We now show that the condition in \cref{eq:56} is satisfied as well.
Let $\vec{x}$ be the original message and $\vec{y}$ the noise added by
the adversary, so that the received vector is
$\vec{z} = \C_tx + \vec{y}$.  Let $\vec{x_0}$ be the output of the
decoder on input $\vec{z}$ computed according to \cref{eq:43}.  We
then have
\begin{flexeq}
  \label{eq:44}
  \norm[\dagger_\mu(k t)]{\C_t(\vec{x}  - \vec{x_0})}
  \leq  \norm[\dagger_\mu(k t)]{\vec{y}} +
  \norm[\dagger_\mu(k t)]{\vec{z} - \C_t\vec{x_0}}
  \eqbreak\leq 2\norm[\dagger_\mu(k t)]{\vec{y}},
\end{flexeq}
where the second inequality follows from \cref{eq:43} since
$\C_t\vec{x}$ is in $\ran{\C_t}$. Applying Theorem~\ref{thm:encoding-matrix}(\ref{item:C-invert}), we now see that there exists a constant $c_0$ such that
\begin{displaymath}
  \norm[\star_\mu(t)]{\vec{x} - \vec{x_0}}
  \leq c_0 t^{-(1/2 - \delta)}\norm[\dagger_\mu(k t)]{\vec{y}},
\end{displaymath}
so that the condition in \cref{eq:56} also holds.

\section{Overview}
\label{sec:preliminaries-and-overview}

This section is devoted to a high-level description of the main ideas
of our construction and its analysis. All main ideas needed for the
proof of \Cref{thm:encoding-matrix} are discussed here, and a roadmap
with forward references to the full arguments is provided.  The
details, being more complicated, have been consigned to
\Cref{sec:progr-atten-gauss,sec:code-matrix}. %

Our starting point is the connection to Euclidean sections of
$(\R^{cT}, \norm[1]{\cdot})$ alluded to in the introduction.
Specifically, we recall the discussion there of rectangular matrices
$A$ whose range is a Euclidean section, or equivalently, which
satisfy \cref{eq:32}.  One standard construction of such a matrix is
to choose a $cT \times T$ random matrix whose entries are i.i.d.\
Gaussian variables.  In order to ensure that
$\norm[2 \rightarrow 2]{A} = O(1)$, it suffices to choose the standard
deviation of the entries to be $\Theta(1/\sqrt{T})$.  For the purposes
of the informal discussion in this section, we will refer to such a
random matrix, whose entries are independent Gaussians with variances
within a constant factor of each other, as a \emph{Dvoretzky matrix}.
The discussion in the introduction showed that a Dvoretzky matrix
suffices if we were interested only in block coding with a block
length of $T$ and did not enforce the online encoding constraint.

The first step to adapting this standard construction to our online
setting is to zero out the entries above the diagonal (in the indexing
of rows and columns introduced in \Cref{sec:rate-functions-lotk}, this
corresponds to enforcing \cref{eq:52}).  However, this is not
sufficient since the entries close to the diagonal are still of order
$O(1/\sqrt{T})$ where $T$ is the total time of transmission.  To see
what the problem is, consider the operation of the encoder and the
decoder at a time $t \ll T$.  In this setting, messages sent by the
encoder up to time $t$ are all attenuated by a factor that is
$O(1/\sqrt{T})$, and this allows the adversary to swamp out the signal
with noise of small $\norm[1]{\cdot}$-norm.  Such a situation will not
allow us to achieve a decoding guarantee similar to \cref{eq:29} where
the guarantee provided at time $t \ll T$ keeps monotonically improving
as the total time $T$ for which the transmission lasts increases (in
fact, in this scenario, the decoding at time $t \ll T$ becomes
progressively worse with increasing
$T$). %
We therefore cannot attenuate all entries of the matrix by a factor of
the form $O(1/\sqrt{T})$; indeed we want entries close to the diagonal
of the matrix to be of order $\tilde{\Omega}(1)$ (so that immediate
decoding is accurate unless there is a noise burst).  On the other
hand, we do want the variances of the matrix entries to have
properties similar to those of Dvoretzky matrices, in the sense that
\begin{enumerate}
\item the sum of variances across a row or column of the matrix is at most a
  constant: intuitively, this is a prerequisite for enforcing that the
  $2\rightarrow 2$-norm of the matrix is a constant, and
\item the sum of their square roots (i.e., standard
  deviations) across a row or column is roughly $\tilde{\Omega}(\sqrt{t})$: intuitively, this
  is a prerequisite for making sure that all vectors in the image of
  the unit $\norm[2]{\cdot}$-ball under the matrix have
  $\norm[1]{\cdot}$-norm about $\tilde{\Omega}(\sqrt{t})$.
\end{enumerate}

To satisfy the above two conditions with the lower triangular
constraint, we consider random matrices whose entries are Gaussians
with progressively attenuated variances.  The construction we actually
use in the proof of \Cref{thm:encoding-matrix} appears in
\Cref{sec:progr-atten-gauss}, but for the purposes of this informal
discussion, we use a slightly simplified version.  Let $k$ be a fixed
constant rate parameter.  We then define the distribution
$\mathcal{A}'_{T, k}$ on \lot{k} matrices such that a $kT \times T$
matrix $M \sim \mathcal{A}'_{T, k}$ is sampled as follows:
\begin{equation}
  \label{eq:1}
  M((i, l), j) = \frac{1}{k} \cdot
  \begin{cases}
    0, & i < j,\\
    g(i - j)\xi_{(i,l), j}, & i \geq j,
  \end{cases}
\end{equation}
where that $\xi_{(i,l),j}$ are independent standard normal random
variables, and
\begin{displaymath}
  g(i) \defeq \frac{1}{\sqrt{i + 1}\log (i + 2)}.
\end{displaymath}
Note that $\sum_{i \geq 0}g(i)^2$ converges, while
$\sum_{i \in [t]}g(i) = \tilde{\Omega}(\sqrt{t})$.  A lower bound on
the probability that $M$ as sampled above has small
$\norm[2 \rightarrow 2]{\cdot}$-norm is established by adapting known
results in the literature: see \Cref{lem:op-norm-of-M}.  The main
technical problem, however, is to show that $\norm[\dagger_\mu(k T)]{M\vec{x}}$ is
large compared to $\norm[\star_\mu(T)]{\vec{x}}$ for all
$\vec{x} \in \R^T$. Again, we emphasize that to prove
\Cref{thm:encoding-matrix}, we actually need to establish this
condition at all times $t \leq T$: however, for now we focus on the
case $t = T$.

We now introduce some notation that will be useful both in our proofs
and in the discussion here (see \Cref{fig:notation} for a pictorial
illustration of the notation introduced here).  For any positive
integer $n$, we define $\tf(n) \defeq \ceil{\lg(n + 1)}$ so that
\tf(n) is the length of the canonical binary representation of $n$.
For a vector $\vec{x} \in \R^T$,
we denote by $\vblock{\vec{x}}{i}$ the sub-vector of $\vec{x}$ of
length $2^{i-1}$ consisting of the entries
$(x_{T + 2 - 2^i}, \cdots, x_{T + 1 - 2^{i-1}})$.  We similarly define
the sub-vector $\cblock{\vec{x}}{j_1, j_2}$ to be the concatenation of
the sub-vectors $\vblock{\vec{x}}{i}$ for $j_2 \leq i \leq j_1$.  When
$j_2 = 1$, we write $\cblock{\vec{x}}{j_1, 1}$ as
$\cblock{\vec{x}}{j_1}$.

Our analysis of $M$ will need to consider the action of appropriate
sub-matrices of $M$ on such sub-vectors; we now introduce notation for
these sub-matrices.  For any matrix $A$
with $T$ columns, let $\mblock{A}{j}$ denote the matrix consisting of
the $2^{j-1}$ columns of $A$ with indices in the interval
$[T + 2 - 2^{j}, T + 1 - 2^{j-1}]$.  We thus have for any such $A$ (in
particular for $M$) that
\begin{displaymath}
  A\vec{x} = \sum_{j = 1}^{\tf(T)} \mblock{A}{j}\vblock{\vec{x}}{j}.
\end{displaymath}
Similarly, we define $\mcblock{A}{j}$ to be the sub-matrix of $A$
consisting of its last $2^j - 1$ columns.  In particular,
$\mcblock{A}{j}$ acts on $\cblock{\vec{x}}{j}$ and we have
\begin{displaymath}
  \mcblock{A}{j}\cblock{\vec{x}}{j}
  = \mblock{A}{j}\vblock{\vec{x}}{j}
  + \mcblock{A}{j - 1}\cblock{\vec{x}}{j - 1}.
\end{displaymath}
We will also need to consider suffixes of the output of these matrices
at several places in the proofs and also in this discussion.
Formally, given an integer $k$ and any matrix $A$, we define
$\lclock{A}{j}$ to be the sub-matrix of $A$ consisting of its last
$k \cdot 2^{j-1}$ rows.  We also extend this notion to vectors in the
co-domain of $A$: for such a vector $\vec{y}$,
$\lclock{\vec{y}}{j}$ denotes the sub-vector consisting of the last
$k\cdot 2^{j-1}$ entries of $\vec{y}$.
\begin{figure*}[t]
  \centering
  \includegraphics[height=0.5\textheight]{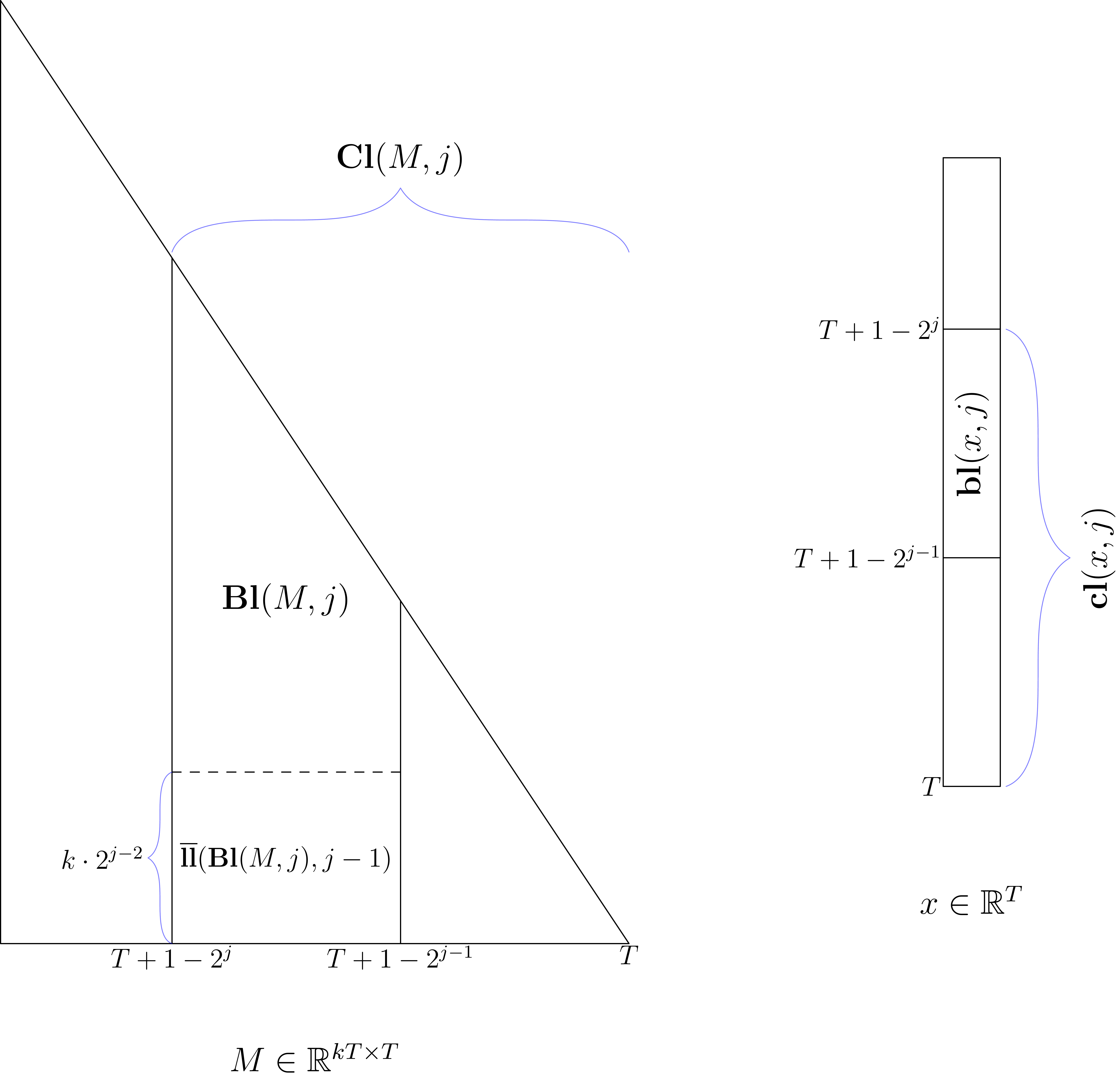}
  \caption{Notation for submatrices and
    sub-vectors\label{fig:notation}}
\end{figure*}

In our proofs, it is easier to work in terms of a matrix $B$ obtained
by rescaling the entries of $M$ in such a way that %
\begin{displaymath}
  \inf_{y \neq 0}\frac{\norm[\dagger_\mu(kT)]{My}}{\norm[\star_\mu(T)]{y}}
  = \sqrt{T}  \cdot \inf_{\norm[2]{\vec{x}} = 1} \norm[1]{Bx}.
\end{displaymath}
Such a \lot{k} matrix $B$ is obtained by setting
\begin{displaymath}
  B_{(i, l), j} =  \frac{M_{(i,l),j}}{\insq{(T- i +
      1)^{1-\mu}(T - j + 1)^\mu}^{1/2}}
\end{displaymath}
for $i \geq j$ and $B_{(i, l), j} = 0$ otherwise.
\Cref{thm:2-1norm} and \Cref{lem:inf-B} then show that for each $j$,
$\lclock{\mblock{B}{j}}{j-1}$ behaves roughly like a Dvoretzky matrix,
in the sense that
$\sup
\frac{\norm[1]{\lclock{\mblock{B}{j}}{j-1}\vec{x}}}{\norm[2]{\vec{x}}}$
is within a factor ${\Theta}(1)$ of
$\inf
\frac{\norm[1]{\lclock{\mblock{B}{j}}{j-1}\vec{x}}}{\norm[2]{\vec{x}}}$,
with probability at least $1 - \exp(-\Omega(k\cdot 2^j))$.
\Cref{cor:inf-B-perturb} strengthens this to show that with the same
probability, the infimum above is not decreased substantially even if
the output of $\lclock{\mblock{B}{j}}{j-1}$ is perturbed with a vector
drawn from a small dimensional subspace (namely, the range of
$\mcblock{B}{j-1}$).

\Cref{lem:perturb-preserve} then shows that the $\norm[1]{\cdot}$ norm
of this perturbation itself is also preserved in the output of
$\lclock{\mblock{B}{j}}{j-1}$.  Together, these results lead to
\Cref{lem:max-norm} which shows, roughly speaking, that with
probability $1 - \exp(-\Omega(k\cdot 2^j))$,
\begin{multline*}
  \norm[1]{\lclock{\mcblock{B}{j}\cblock{\vec{x}}{j}}{j-1}}\\
  \geq \max\left\{
    \norm[1]{\lclock{\mblock{B}{j}\vblock{\vec{x}}{j}}{j-1}},\notbool{j2col}{}{
    \right.\\
  \left.
  }(1 - \epsilon)\norm[1]{\lclock{\mcblock{B}{j - 1}\cblock{\vec{x}}{j -
          1}}{j-1}}\right\},
\end{multline*}
for some small constant $\epsilon > 0$.  Informally, this says that
the output of each trailing principal sub-matrix $\mcblock{B}{j}$
preserves both the output of its left most block $\mblock{B}{j}$, as
well as the output of the remaining trailing principal sub-matrix
$\mcblock{B}{j-1}$.  Underlying these results is a sequence of
$\epsilon$-net arguments, which use concentration bounds on the
$\ell_1$ norms of Gaussian vectors with independent entries of
non-identical means and variances, provided in
\Cref{cor:conc,thm:l1-concentration}. %

Finally, \Cref{thm:main-lower-b,cor:invert-M} use \Cref{lem:max-norm} in an 
induction to show that with probability at least
$1 - \exp(-\Omega(k))$,
\begin{equation}
  \norm[\dagger_\mu(k T)]{M\vec{x}} \geq \tilde{\Omega}(T^{1/2 -
    \delta})\norm[\star_\mu(T)]{\vec{x}}\label{eq:50}
\end{equation}
\Cref{eq:50} establishes that $M$ has the requisite properties at time
$T$, but recall that our code requires online decoding at all times
$t \leq T$.  Unfortunately, we cannot take an union bound over all $t$
using \cref{eq:50} unless we choose $k = \Omega(\log T)$, which would
lead to a very low communication rate of $1/k = O(1/\log T)$ (recall
that what we actually want, and achieve, is $k$ a constant).

However, there is a simple remedy if we would be willing to carry
information only about a sufficiently delayed prefix of $\vec{x}$.  In
particular, \Cref{thm:main-lower-b,cor:invert-M} also show that if
$j_0$ is chosen so that $j_0 = \Omega(\log \log T)$, then $M$ carries
enough information about $\cblock{\vec{x}}{\tau_0, j_0}$ (recall that
this is the prefix of $\vec{x}$ which ignores its last
$2^{j_0 - 1} - 1 = \poly{\log T}$ entries) so that with probability at least
$1 - O(1/T^2)$:
\begin{displaymath}
  \norm[\dagger_\mu(k T)]{M\vec{x}} \geq O(T^{1/2 -
    \delta})\norm[\star_\mu(T)]{\cblock{\vec{x}}{\tau_0, j_0}}.
\end{displaymath}
We can now indeed take a union bound over all $t \leq T$ to see that
the above is true at all times $t \leq T$
with probability at least $1 - O(1/T)$.  However, the price we pay for
this is that we cannot say anything about the most recent $\poly{\log
T}$ characters %
of the message.  This can be fixed by making the code
systematic: the details are in \Cref{sec:code-matrix}.

We emphasize here that although the above discussion often refers to
the total time of communication $T$, our actual construction does
not assume a knowledge of $T$.  In particular, the rate and error
guarantees in \Cref{thm:encoding-matrix} are achieved also at times
$t$ that might be much smaller than the eventual total time $T$.

The rest of the paper is devoted to the details of the proof.

\section{Progressively attenuated Gaussian matrices}
\label{sec:progr-atten-gauss}

Our proof of \Cref{thm:encoding-matrix} will proceed through an
analysis of a specific distribution over random \lot{k}
matrices.  We start by recalling some results from the literature that
will be used in our proofs.

\subsection{Technical preliminaries}
\label{sec:techn-prel}

\subsubsection{Operator norm of Gaussian matrices}
\label{sec:oper-norm-gauss}

We will use the following result on the operator norm of matrices
with independent Gaussian entries.
\begin{theorem}[\textbf{Bandeira and van Handel~\cite[Theorem
    3.1]{bandeira_sharp_2014}}]
  \label{thm:expected-matrix-norm}
  Let $A$ be a $n \times m$ random matrix with independent mean zero
  Gaussian entries such that $A_{ij} \sim \mathcal{N}(0, a_{ij}^2)$.
  Then
  \begin{displaymath}
    \E{\norm[2\rightarrow 2]{A}} \leq \frac{3}{2}\inp{\sigma_1 + \sigma_2
      + 10 \sigma_0 \sqrt{\log \min\inp{n ,m}}},
  \end{displaymath}
  where
  \begin{displaymath}
    \sigma_1^2 \defeq \max_{i \in [n]}\sum_{j = 1}^m
    a_{ij}^2; \;
    \sigma_2^2 \defeq \max_{j \in [m]}\sum_{i = 1}^n
    a_{ij}^2; \;
    \sigma_0 \defeq \max_{\substack{i \in [n]\\j \in [m]}}\abs{a_{ij}}.
  \end{displaymath}
\end{theorem}

\subsubsection{$\epsilon$-nets}
\label{sec:epsilon-nets}
We will use the following standard facts about $\epsilon$-nets for
subspaces of $(\R^n, \norm[p]{\cdot})$ for $p \geq 1$ (see, e. g.,
\cite{figiel_dimension_1977}).
\begin{fact}
  \label{fct:epsilon-nets}
  Let $U$ be a subspace of $(\R^n, \norm[p]{\cdot})$ of dimension at
  most $d$.  Then, for $\epsilon \leq r$ the $\norm[p]{\cdot}$-ball
  (respectively, the $\norm[p]{\cdot}$-sphere) of radius $r$ in $U$
  has an $\epsilon$-net in $\norm[p]{\cdot}$ of size at most
  $(3r/\epsilon)^d$.
\end{fact}
\begin{fact}
  \label{fct:op-norm}
  Let $p, q \geq 1$, and let $M$ be a $m \times n$ real matrix.  If
  $\norm[q]{M\vec{x}} \leq c$ for all $\vec{x}$ in a $\norm[p]{\cdot}$
  $(1/2)$-net of the $\norm[p]{\cdot}$ sphere in $\R^n$ , then
  $\norm[p \rightarrow q]{M} \leq 2c$.
\end{fact}

\subsection{The distribution $\mathcal{A}_{T, k}$}
\label{sec:distr-mathc-k}
We now describe the distribution on random $k$-lower triangular
matrices that will be used in the proof of \Cref{thm:encoding-matrix}.

Let $T$ be a positive integer, and set $\tau = \tf(T)$, where
\[\tf(T) \defeq \ceil{\lg (T + 1)}\] is the number of bits in the
canonical binary representation of $T$ .  Given a rate parameter $k$,
$\mathcal{A}_{T,k}$ is a distribution on $kT \times T$ \lot{k}\
matrices, such that a matrix \(M \sim \mathcal{A}_{T, k}\) is sampled
as follows:
\begin{equation}
  \label{eq:8}
  M((i, l), j) = \frac{1}{k\cdot\tf(i)^4} \cdot
  \begin{cases}
    0, & i < j,\\
   \frac{1}{\sqrt{i - j + 1}}\cdot\xi_{(i,l), j}, & i \ge j.
  \end{cases}
\end{equation}
where $i, j \in [T], l \in [k]$, and the \(\xi_{(i,l), j}\) are
independent standard normal random variables.  Note that we divide the
rows of $M$ into $T$ \emph{segments}, where the $i$th segment is of
size $k$, and index the rows by a pair $(i,l)$ where $i \in [T]$
denotes the segment, and $l \in [k]$ determines the offset in
the segment.
\begin{remark}
  Note that the distribution $\mathcal{A}_{T, k}$ is
  ``time-invariant'' in the sense that for any $1 \leq t \leq T$, the
  $kt \times t$ leading principal submatrix of $M$ sampled from
  $\mathcal{A}_{T, k}$ is also a faithful sample from
  $\mathcal{A}_{t, k}$.
\end{remark}

\subsection{The distribution of $\norm[2\rightarrow 2]{M}$} We begin
with a short discussion of the operator norm of $M$ sampled according
to $\mathcal{A}$; much of our technical work would be devoted to the
study of $M^{-1}$.  For the operator norm, however, the following
corollary of \Cref{thm:expected-matrix-norm} of Bandeira and van
Handel will be sufficient for our purposes.
\begin{lemma}
  \label{lem:op-norm-of-M}
  For any $\gamma \in (0, 1)$ there exists a positive integer $c_0$
  such that if $k > c_0$ then $M \sim \mathcal{A}_{T, k}$ satisfies
  $\norm[2\rightarrow 2]{M} \leq 1$ with probability at least
  $1 - \gamma$.
\end{lemma}
\begin{proof}
  For $1 \leq i \leq \tf(T)$, let $M(i)$ denote the submatrix of $M$
  consisting of the consecutive rows from $(2^{i-1}, 1)$ to
  $(\min\inp{2^i - 1, T}, k)$. Here, we
  are using the indexing scheme for rows of $M$ that was defined in
  \cref{eq:8}.  Note that the number of non-zero columns of $M(i)$ is
  at most $2^{i} - 1$.  We now apply \Cref{thm:expected-matrix-norm}
  to each $M(i)$.  In the notation of that theorem, we have for $M(i)$
  \begin{displaymath}
    \sigma_1 \leq \frac{1}{k \cdot
      i^3}\text{, }
    \sigma_2 \leq \frac{1}{\sqrt{k} \cdot
      i^3}\text{ and }
    \sigma_0 \leq \frac{2}{k\cdot i^4},
  \end{displaymath}
  where we use the estimate
  \begin{displaymath}
    \sum_{x=1}^n\frac{1}{x} \leq \sum_{x = 1}^{2^{\tf(n)} - 1}
    \frac{1}{x} \leq \sum_{j=1}^{\tf(n)} \sum_{x : \tf(x) =
      j}\frac{1}{2^{j - 1}} = \tf(n).
  \end{displaymath}
  The theorem then implies that when $k \geq c$ for
  $c = c(\gamma)$ large enough, we have
  $\E{\norm[2\rightarrow 2]{M(i)}} \leq \frac{\gamma}{4i^3}$ for all
  $1 \leq i \leq \tf(T)$.  Thus, for $1 \leq i \leq \tf(T)$,
  \begin{displaymath}
    \Pr{\norm[2\rightarrow 2]{M(i)} > \frac{1}{i\sqrt{2}}} \leq
    \frac{\gamma}{2i^2}.
  \end{displaymath}
  By a union bound (and using $\sum_{i \geq 1} (1/i^2) < 2$), we get
  that with probability at least $1 - \gamma$
  \begin{equation}
    \norm[2 \rightarrow 2]{M(i)} \leq \frac{1}{i\sqrt{2}}\text{ for
      all $1 \leq i \leq \tf(T)$}. \label{eq:46}
  \end{equation}
  When the event in \cref{eq:46} occurs, we have
  $\norm[2\rightarrow 2]{M} \leq 1$, since for any $\vec{x} \in \R^T$
  (here $\vec{x}_{[l]}$ denotes the prefix of $\vec{x}$ consisting of
  its first $l$ co-ordinates)
  \begin{flexeq*}
    \norm[2]{M\vec{x}}^2
    = \sum_{i = 1}^{\tf(T)}\norm[2]{M(i)\vec{x}_{[2^i - 1]}}^2
    \eqbreak\leq \sum_{i = 1}^{\tf(T)}\norm[2 \rightarrow 2]{M(i)}^2 \norm[2]{\vec{x}_{[2^i - 1]}}^2
    \leq \sum_{i = 1}^{\tf(T)}\frac{1}{2i^2} \norm[2]{\vec{x}_{[2^i -
    1]}}^2
    \eqbreak \leq \sum_{j=1}^Tx_j^2\sum_{i=\tau(j)}^{\tf(T)}\frac{1}{2i^2} \leq \norm[2]{\vec{x}}^2,
  \end{flexeq*}
  where the last inequality uses $\sum_{i\geq 1}(1/i^2) < 2$.
\end{proof}

\subsection{Invertibility of $M$}
\label{sec:invertibility-a}
To ease notation, we fix a $\mu \in (0, 1]$ in the rest of this
section, and proceed to study the robust invertibility of a matrix $M$
sampled from $\mathcal{A}_{T, k}$ with respect to the
$\norm[\dagger_\mu(k T)]{\cdot}$ and $\norm[\star_\mu(T)]{\cdot}$
norms by analyzing the quantity
$\inf_{y \neq
  0}\frac{\norm[\dagger_\mu(kT)]{My}}{\norm[\star_\mu(T)]{y}}$.  {The
constants appearing in the statements of the theorems appearing below
therefore carry an implicit dependence upon this fixed value of $\mu$.}

Our first step is to pass to standard unweighted norms via a simple
reduction.  Let $L$ be a $kT \times kT$ diagonal matrix with
$L_{i,i} = (k T)^{(1-\mu)/2} \cdot {(kT - i + 1)}^{-(1 - \mu)/2}$, and
let $R$ be a $T \times T$ diagonal matrix with
$R_{i, i} = T^{\mu/2}\cdot(T- i + 1)^{-\mu/2}$.  We then have
\begin{equation}
  \label{eq:24}
  \inf_{y \neq 0}\frac{\norm[\dagger_\mu(kT)]{My}}{\norm[\star_\mu(T)]{y}} =
  \inf_{x \neq 0} \frac{\norm[1]{LMRx}}{\norm[2]{x}} \geq
  \sqrt{T}  \cdot \inf_{\norm[2]{\vec{x}} = 1} \norm[1]{Bx},
\end{equation}
where the matrix $B$ is defined in terms of $M$ as follows:
\begin{equation}
  \label{eq:def-B-from-M}
  B_{(i, l), j} =  \frac{M_{(i,l),j}}{\insq{(T- i +
      1)^{1-\mu}(T - j + 1)^\mu}^{1/2}}.
\end{equation}
Denote the distribution of $B$ obtained from
$M \sim \mathcal{A}_{T, k}$ as $\mathcal{B}_{T, k}$.  We will now
study the properties of the blocks $\mblock{B}{j}$ for $B$ sampled
from this distribution in detail.  We start with an investigation of
their $2 \rightarrow 1$ norm.

\newcommand{\kt}[1]{\ensuremath{k(\tf(#1))}}
\begin{theorem}
  \label{thm:2-1norm}
  Let $k$ be a rate parameter such that $k \geq 1 + \log 6$, and let
  $B \sim \mathcal{B}(T, k)$ for some positive integer $T$.  Then for
  each $2 \leq j \leq \tf(T)$,
  \begin{flexeq*}
    \Pr{\norm[2\rightarrow 1]{\lclock{\mblock{B}{j}}{j - 1}} >
      \frac{256}{\sqrt{\mu}\tf(T)^4}} \eqbreak \leq
    \exp\inp{-(k - \log 6) 2^{j - 1}}.
  \end{flexeq*}
\end{theorem}
\begin{proof}
  Fix $2 \leq j \leq \tf(T)$, and let $S$ be any $1/2$-net for the
  unit sphere in $(R^{2^{j-1}}, \norm[2]{\cdot})$.  Note that we can
  choose $S$ so that $\abs{S} \leq \exp\inp{2^{j-1}\log 6}$. For ease
  of notation, we index the co-ordinates of any $\vec{x} \in S$ from
  $2^{j-1}$ to $2^{j} - 1$.  Now, for any such $\vec{x} \in S$, we
  have
  \begin{equation}
    \norm[1]{\lclock{\mblock{B}{j}\vec{x}}{j - 1}} =
    \sum_{i = 1}^{2^{j-2}}\sum_{l=1}^{k}\abs{X_{i,l}},\label{eq:5}
  \end{equation}
  where $X_{i,l}$ are independent mean zero normal variables with
  variances
  \begin{equation}
    \sigma_{i,l}^2 = \frac{1}{k^2\cdot \tf(T - i +
      1)^8}\sum_{s = 2^{j - 1}}^{2^j - 1}\frac{x_s^2}{s^\mu \cdot i^{1-\mu}\cdot(s - i + 1)}.
    \label{eq:57}
  \end{equation}
  Note that since $j \leq \tf(T)$, we have $T \geq 2^{j-1}$.  For
  $1 \leq i \leq 2^{j-2}$, this implies that
  $\tf(T - i + 1) \geq \max(1, \tf(T) - 1) \geq \tf(T)/2$, so that we have
  \begin{align}
    \sum_{i=1}^{2^{j - 2}}\sum_{l=1}^{k}
    \sigma_{i,l}^2
    &\leq \frac{2^{8}}{k \tf(T)^8} \sum_{i=1}^{2^{j-2}}
      \sum_{s = 2^{j - 1}}^{2^{j}-1}
      \frac{x_s^2}{s^\mu \cdot i^{1-\mu}\cdot(s - i + 1)}\nonumber\\
    &\leq \frac{2^{11}\cdot 2^{-(1 + \mu)j} \cdot {\norm[2]{x}^2}}{k \tf(T)^8}
        \cdot \sum_{i=1}^{2^{j-2}} \frac{1}{i^{1-\mu}}\nonumber\\
    &\leq \frac{2^{11}\cdot 2^{-j}}{\mu k \tf(T)^8} \norm[2]{\vec{x}}^2,\label{eq:58}
  \end{align}
  where in the first inequality we use $0 < \mu \leq 1, s \geq
  2^{j-1}$ and $s - i + 1 \geq 2^{j-2}$, and in the second inequality
  the fact that $\sum_{i=1}^Ni^{\mu - 1} \leq N^\mu/\mu$.
  
  We now apply \Cref{lem:gauss-l1-upper} to the sum in \cref{eq:5}.
  The number of terms $n$ is $k\cdot 2^{j-2}$, and we set the
  parameter $\alpha$ in \Cref{lem:gauss-l1-upper} to
  $\alpha = \frac{2^7}{\sqrt{\mu}\tf(T)^4\sqrt{n}}\norm[2]{\vec{x}}$ to get
  \begin{flexeq}
    \Pr{\norm[1]{\lclock{\mblock{B}{j}\vec{x}}{j - 1}}
      > \nfrac{2^7}{(\sqrt{\mu}\tf(T)^4)}\norm[2]{\vec{x}}} \eqbreak\leq
    \exp\inp{-k \cdot 2^{j}}.
  \end{flexeq}
  A union bound over all $\vec{x} \in S$ now yields
  \begin{flexeq}
    \Pr{\exists \vec{x} \in S,
      \norm[1]{\lclock{\mblock{B}{j}\vec{x}}{j - 1}} >
      \nfrac{2^7}{(\sqrt{\mu}\tf(T)^4)}\norm[2]{\vec{x}}}
    \eqbreak\leq \exp\inp{-(k - \log 6)\cdot 2^{j}}.\label{eq:18}
  \end{flexeq}
  Since $S$ is a $1/2$-net, \Cref{fct:op-norm} implies that
  $\norm[2\rightarrow 1]{\lclock{\mblock{B}{j}}{j - 1}} \leq
  2^8/(\sqrt{\mu}\tf(T)^4)$ with probability at least
  $ 1 - \exp\inp{-(k - \log 6)\cdot 2^{j-1}}$.
\end{proof}

The next lemma shows that the $\norm[1]{\cdot}$ norm of the output of
$\mblock{B}{j}$ cannot be very small. %
\begin{lemma}
  \label{lem:inf-B}
  There exist positive constants $c_1, c_2 > 0$ such that the
  following is true.  For any integer $j \geq 2$, any $k \geq c_1$,
  and any vector $\vec{y}$,
  \begin{flexeq*}
    \mathbb{P}\left[
      \exists x, \norm[2]{\vec{x}} = 1 \text{ and } \right.
    \eqbreak \left.
      \norm[1]{\lclock{\mblock{B}{j}\vec{x} + \vec{y}}{j - 1}} < c_2/(\tf(T)^4)
    \right]
    \eqbreak \leq \exp\inp{-c_2 k \cdot 2^{j}}.
  \end{flexeq*}
\end{lemma}
\begin{proof}
  Let $S$ be an $\epsilon$-net for the unit sphere in
  $\inp{\R^{2^{j-1}}, \norm[2]{\cdot}}$, for an $\epsilon$ to be
  determined later.  As in \cref{eq:5} in the proof of
  \Cref{thm:2-1norm}, for any $\vec{x} \in S$, we have
  \begin{equation}
    \label{eq:10}
    \norm[1]{\lclock{\mblock{B}{j}\vec{x} + \vec{y}}{j - 1}}
    = \sum_{i = 1}^{2^{j-2}}\sum_{l=1}^{k}\abs{X_{i,l} + y_{i, l}}
  \end{equation}
  where $X_{i,l}$ are independent mean zero normal variables with
  variances $\sigma_{i,l}$ as defined in \cref{eq:57}.  Recall also
  that $\tf(T) - 1 \leq \tf(T - i + 1) \leq \tf(T)$ since
  $i \leq 2^{j-2}$.

  Since we are interested in upper bounding the probability that the
  above sum is small, it follows from \Cref{cor:stochastic-domination}
  that the worst case is $\vec{y} = 0$. %
  In preparation to apply \Cref{cor:conc} to the above sum with
  $\vec{y} = 0$, we also note that since the $\sigma_{i,l}^2$ are
  positive linear functions of the $x_s^2$ and
  $\norm[2]{\vec{x}} = 1$,
  $GM\inp{\inp{\sigma_{i,l}}_{\substack{i \in [2^{j-2}]\\l \in
        [k]}}}$ is minimized when $\vec{x} = \vec{e_s}$
  for some $s \in [2^{j-1}, 2^j - 1]$.  We thus have
  \begin{align*}
    \notbool{j2col}{}{&}
    GM\inp{\inp{\sigma_{i,l}}_{\substack{i \in [2^{j-2}]\\l \in [k]}}}
    \notbool{j2col}{}{\\}
                      &\geq \frac{1}{k\tf(T)^4}
                        \cdot \frac{1}{\sqrt{GM\inp{(i^{1-\mu})_{i \in [2^{j-2}]}}}}
                        \notbool{j2col}{\cdot}{\\ &\qquad \cdot}
                        \min_{s \in [2^{j-1}, 2^j - 1]}
     \frac{1}{\sqrt{s^\mu GM\inp{(s - i + 1)_{i \in [2^{j-2}]}}}}\\
   &\geq  \frac{2^{-j/2}}{k\tf(T)^4} \min_{s \in [2^{j-1}, 2^j - 1]}
     \frac{1}{\sqrt{s}}.
  \end{align*}
  Applying \Cref{cor:conc}, and noting that the number of terms in the
  sum in \cref{eq:10} for which the geometric mean was taken above is
  $k\cdot 2^{j-2}$, we now find a positive constant $c > 0$ such that the
  following holds for all $\tau \in (0, 1)$:
  \begin{equation*}
    \Pr{\norm[1]{\lclock{\mblock{B}{j}\vec{x} + \vec{y}}{j - 1}}
      < \frac{c\tau}{\tf(T)^4}}
    \leq \tau^{k\cdot 2^{j-2}}.
  \end{equation*}
  Taking a union bound over all $\vec{x}$ in the $\epsilon$-net
  $S$, and then using the bound on
  $\norm[2\rightarrow 1]{\lclock{\mblock{B}{j}}{j-1}}$ derived in the
  proof of \Cref{thm:2-1norm}, we then have
  \begin{multline}
    \mathbb{P}\Big[ \exists x, \norm[2]{x} = 1\text{, and } \eqbreak
    \left. \norm[1]{\lclock{\mblock{B}{j}\vec{x} + \vec{y}}{j - 1}}
      < (c\tau - 256\mu^{-1/2}\epsilon)/\tf(T)^4
    \right] \\
    \leq \exp\inp{-2^{j-1}\cdot \insq{(k/2)\log (1/\tau) - \log
        (3/\epsilon)}} \eqbreak + \exp\inp{-2^{j-1}\cdot(k - \log
      6)}.
    \label{eq:2}
  \end{multline}
  Since $k \geq c_1$ for a large enough $c_1$, the claim now follows
  after choosing $\epsilon$ and $\tau$ to be appropriate constants.
\end{proof}

We now consider small dimensional perturbations to the output of
$\mblock{B}{j}$ for $j \geq 2$, and start with a corollary of
\Cref{lem:inf-B}.
\begin{corollary}
  \label{cor:inf-B-perturb}
  There exist positive constants $C, C_1, C_2 > 0$ such that the
  following is true.  For any $j \geq 2$, $k \geq C_1$, and $V$ an
  arbitrary subspace of dimension at most $2(2^{j}-1)$,
  \begin{flexeq*}
    \mathbb{P}\left[
      \exists \vec{x} \in \R^{2^{j-1}}, y \in V \text{ s. t. }
      \norm[1]{\lclock{\mblock{B}{j}\vec{x} + \vec{y}}{j - 1}} \right.
    \eqbreak < \left.
      C\norm[2]{\vec{x}}/(\tf(T)^4)
    \right]
    \leq \exp\inp{-C_2 k \cdot 2^{j}}.
  \end{flexeq*}
\end{corollary}
\begin{proof}
  Let $U$ be the vector space
  $\inb{\lclock{\vec{y}}{j - 1} \st \vec{y} \in V}$.  Note that we can
  replace $V$ by $U$ in the statement of the corollary (i.e., if the
  result holds for $U$, then it also holds for $V$).  We therefore
  restrict our attention to $U$.  Note that the dimension of $U$ is no
  more than the dimension of $V$.

  Let $c_2$ be as in \Cref{lem:inf-B} and define $C = c_2/2$.  From
  \Cref{thm:2-1norm} we know that
  $\norm[2\rightarrow 1]{\lclock{\mblock{B}{j}}{j-1}} \leq
  256/(\sqrt{\mu}\tf(T)^4)$ with probability at least
  $1 - \exp\inp{-\Theta(k \cdot 2^{j})}$.  Under this event we also
  have
  \[\norm[1]{\lclock{\mblock{B}{j}\vec{x} + \vec{y}}{j - 1}} \geq
  C/\tf(T)^4\] whenever $\norm[2]{\vec{x}} = 1$ and
  $\norm[1]{\vec{y}} > (C + 256\mu^{-1/2})/\tf(T)^4$.

  Therefore, let $N$ be a $\inp{C/\tf(T)^4}$-net in $\ell_1$ for the
  $\ell_1$ ball of radius $(C + 256\mu^{-1/2})/\tf(T)^4$ in $U$.  We have
  $\abs{N} \leq \exp\inp{c'2^{j}}$ for some $c' > 0$.  Thus, applying
  \Cref{lem:inf-B} to each element in $N$ and then taking a union
  bound, we have
  \begin{flexeq*}
    \mathbb{P}\left[
      \exists \vec{x} \in \R^{2^{j-1}}, y \in N \text{ s. t. }
      \norm[2]{\vec{x}} = 1
      \text{ and } \right.
      \eqbreak \left. \norm[1]{\lclock{\mblock{B}{j}\vec{x} + \vec{y}}{j - 1}}
      < 2C/\tf(T)^4
    \right]
    \eqbreak \leq \exp\inp{-c''k \cdot 2^{j}},
  \end{flexeq*}
  for some positive constant $c''$ whenever $k \geq c_1$ for some
  other positive constant $c_1$. Using the fact that $N$ is a
  $(C/\tf(T)^4)$-net in $\ell_1$ we get the claimed result.
\end{proof}

We now show that adding the output of $\mblock{B}{j}$ does not shrink
the size of the perturbation either, as long as the perturbations
comes from a small dimensional space.
\begin{lemma}
  \label{lem:perturb-preserve}
  For any $\gamma \in (0, 1)$ there exist positive constants
  $c_1 = c_1(\gamma), c_2 = c_2(\gamma)$ such that for any integers
  $j \geq 2$ and any $k > c_1$, the following is true.  Let $V$ be an
  arbitrary subspace of dimension at most $2(2^{j}-1)$.  Then,
  \begin{flexeq*}
    \mathbb{P}\left[
      \exists \vec{x} \in \R^{2^{j-1}}, y \in V \text{ s. t. } \right.
    \eqbreak \left.
      \norm[1]{\lclock{\mblock{B}{j}\vec{x} + \vec{y}}{j - 1}}
      < \gamma\norm[1]{\lclock{\vec{y}}{j-1}}
    \right]
    \eqbreak \leq \exp\inp{-c_2k \cdot 2^{j}}.
  \end{flexeq*}
\end{lemma}
\begin{proof}
  From \Cref{thm:2-1norm} and \Cref{cor:inf-B-perturb} we have that
  for some constant $c > 0$, the following events occur with
  probability at least $1 - \exp\inp{-\Theta(k \cdot 2^j)}$ for all
  large enough constant $k$:
  \begin{enumerate}
  \item
    $\frac{\norm[1]{ \lclock{\mblock{B}{j}\vec{x} }{j - 1}} }{
      \norm[2]{\vec{x}} } \leq 256\mu^{-1/2}/\tf(T)^4$, for all $\vec{x} \neq 0$, and
    \vspace{0.2\baselineskip}
  \item
    $\norm[1]{\lclock{\mblock{B}{j}\vec{x} + \vec{y}}{j - 1}} \geq
    c\norm[2]{\vec{x}}/\tf(T)^4$ for all
    $\vec{x} \in \R^{2^{j-1}}, \vec{y} \in V$.
  \end{enumerate}
  We assume henceforth that both the above events occur.  In
  particular, we have
  \begin{flexeq}
    \label{eq:19}
    \norm[1]{\lclock{\vec{y}}{j-1}} = 1, \quad
    \norm[2]{\vec{x}} \geq \frac{\gamma}{c}\cdot \tf(T)^4
    \notbool{j2col}{\qquad}{\;}
    \eqbreak \implies
    \notbool{j2col}{\qquad}{\;}
    \norm[1]{\lclock{\mblock{B}{j}\vec{x} + \vec{y}}{j - 1}}
    \geq \gamma \norm[1]{\lclock{\vec{y}}{j-1}}.
  \end{flexeq}
  Let $U$ be the vector space
  $\inb{\lclock{\vec{z}}{j - 1} \st \vec{z} \in V}$.  Now, let $N_{z}$
  be an $\epsilon$-net in $\ell_1$ for the set
  $\inb{\vec{z} \in U \st \norm[1]{\vec{z}} = 1}$ for
  $\epsilon = (1-\gamma)/(1 + 257\mu^{-1/2})$, and $N_{x}$ and $\epsilon_1$-net in
  $\ell_2$ for the $\ell_2$-ball of radius
  $\frac{\gamma}{c} \cdot \tf(T)^4$ in $\R^{2^{j-1} - 1}$ for
  $\epsilon_1 = \epsilon \cdot \tf(T)^4$.  $N_{z}$ and $N_{x}$ can be
  chosen so that
  $\abs{N_{z}}\cdot\abs{N_{x}} \leq \exp\inp{c'\cdot2^{j}}$
  where $c' = c'(\gamma) > 0$.  Let
  $\gamma' = \gamma + (1 + 256\mu^{-1/2})\epsilon < 1$.  We now have, for any
  $\vec{z} \in N_{z}$ and $\vec{x} \in N_{x}$,
  \begin{equation}
    \label{eq:20}
    \norm[1]{\lclock{\mblock{B}{j}\vec{x}  + \vec{z}}{j - 1}} =
    \sum_{i = 1}^{2^{j-2}} \sum_{l=1}^{k} \abs{X_{i,l} + z_{i, l}},
  \end{equation}
  where, as before, $X_{i,l}$ are independent mean zero normal
  variables with variances $\sigma_{i, l}$ as in \cref{eq:57}.
  In preparation to apply \Cref{thm:l1-concentration}, we now estimate
  \begin{align*}
    \sum_{i=1}^{2^{j-2}}\sum_{\ell = 1}^k  \sigma_{i,l}
    &\geq \frac{1}{\tf(T)^4} \sum_{i=1}^{2^{j-2}}
      \sqrt{
      \sum_{s = 2^{j-1}}^{2^j - 1}
      \frac{x_s^2}{s^\mu i^{1-\mu} (s - i + 1)}
      }\\
    &\geq \frac{1}{\norm[2]{\vec{x}}\tf(T)^4} \sum_{s = 2^{j-1}}^{2^j - 1} \frac{x_s^2}{s^{\mu/2}}
      \sum_{i=1}^{2^{j-2}}  \frac{1}{i^{(1-\mu)/2}\sqrt{s - i + 1}}\\
    &\geq \frac{\norm[2]{\vec{x}}}{\tf(T)^4} 2^{-j(1+\mu)/2}
      \sum_{i=1}^{2^{j-2}} \frac{1}{i^{(1-\mu)/2}}\\
    &\geq c_0 \text{ for some fixed constant $c_0(\gamma, \mu)$}.
  \end{align*}
  Here, the second inequality uses the concavity of the square root
  function, the last that $\vec{x} \in N_x$ so that
  $\norm[2]{\vec{x}} = \gamma\tf(T)^4/c$, and the rest are elementary
  estimates.  Now, using the upper bound on $\sum \sigma_{i,l}^2$
  obtained in \cref{eq:58} (while remembering that the vector
  $\vec{x}$ in that calculation needs to be scaled to have length
  $\gamma\tf(T)^4/c$ instead of $1$), we can apply
  \Cref{thm:l1-concentration} to get that for some constant
  $c'' = c''(\gamma) > 0$,
  \begin{flexeq*}
    \Pr{\norm[1]{\lclock{\mblock{B}{j}\vec{x}}{j - 1}  + \vec{z}} <
      \gamma'\norm[1]{z}}
    \eqbreak \leq \exp\inp{-c''k \cdot 2^{j}}.
  \end{flexeq*}
  Taking a union bound over the product $N_{z} \times N_{x}$ of the
  two nets, using \cref{eq:19} and recalling that
  $\gamma' = \gamma + (1 + 256\mu^{-1/2})\epsilon$ and that $N_{z}$
  and $N_{x}$ are $\epsilon$ and $\epsilon\cdot \tf(T)^4$ nets
  respectively, we deduce that for some constant $D = D(\gamma) > 0$
  \begin{flexeq*}
    \mathbb{P}\left[
      \exists \vec{x} \in R^{2^{j-1}}, \vec{z} \in U
      \text{ s.t. } \norm[1]{\vec{z}} = 1
      \text{ and } \right.
    \eqbreak \left.
      \norm[1]{\lclock{\mblock{B}{j}\vec{x}}{j -1} + \vec{z}} < \gamma
    \right]
    \eqbreak \leq \exp\inp{-Dk\cdot 2^j}
  \end{flexeq*}
  when $k \geq c_1$ for $c_1$ large enough.  The result now follows.
\end{proof}
Combining the results of
\Cref{cor:inf-B-perturb,lem:perturb-preserve}, we get
\begin{lemma}
  \label{lem:max-norm}
  There exists a positive constant $C$ such that for any
  $\gamma \in (0, 1)$ there exist positive constants
  $c_1 = c_1(\gamma), c_2 = c_2(\gamma)$ such that for any integers
  $j \geq 2$ and any $k \geq c_1$, the following is true.  Let $V$ be
  an arbitrary subspace of dimension at most $2^{j - 1}-1$.  Then,
  \begin{flexeq*}
    \mathbb{P}\left[
      \exists \vec{x} \in \R^{2^{j-1}}, y \in V \text{ s. t. }
      \norm[1]{\lclock{\mblock{B}{j}\vec{x} + \vec{y}}{j - 1}} \right.
    \eqbreak < \left.
      \max\inb{C\norm[2]{\vec{x}}/\tf(T)^4,
        \gamma\norm[1]{\lclock{\vec{y}}{j-1}}}
    \right] \\
    \leq \exp\inp{-c_2k\cdot 2^{j}}.
  \end{flexeq*}
\end{lemma}

We are now ready to prove the main theorem of this section.
\begin{theorem}
  \label{thm:main-lower-b}
  For any $\kappa \in (0, 1)$ there exist positive constant $c_0, c_1$
  and $c_2$ such that the following is true.  Let $T$ be any positive
  integer and set $\tau \defeq \tf(T)$.  For any rate parameter
  $ k \geq c_1$ and $j_0 \in [1, \tau]$,
  \begin{flexeq*}
    \mathbb{P}_{B \sim \mathcal{B}(T, k)}
    \left[
      \exists \vec{x} \in \R^T \text{ s. t.}
      \norm[1]{\lclock{B\vec{x}}{\tau}} \right.
    \eqbreak < \left. c_0\kappa^{\tau-j_0}\norm[2]{
      \cblock{\vec{x}}{\tau, j_0}}/\tau^4
    \right]
    \eqbreak \leq \exp\inp{-c_2k \cdot 2^{j_0-1}}.
  \end{flexeq*}
\end{theorem}

\begin{proof}[Proof of \Cref{thm:main-lower-b}]
  Let $c_0(\kappa)$ be a fixed constant to be determined later.  For
  $j_0 \leq j \leq \tau$, let $\mathcal{E}_j$ be the event that
  \begin{flexeq*}
    \norm[1]{\lclock{\mcblock{B}{j}\vec{x}}{j}}
    \geq c_0\kappa^{j-j_0} \norm[2]{\cblock{\vec{x}}{j, j_0}}/\tau^4
    \eqbreak\quad \forall \vec{x} \in \R^{2^{j} - 1}.
  \end{flexeq*}
  \Cref{cor:inf-B-perturb} (or, in the case $j_0 = 1$, a direct
  calculation identical to that in \Cref{lem:inf-B}) shows that if
  $c_0$ is a small enough positive constant, there exist positive
  constants $c', c''$ (independent of $j_0)$ such that
  $\Pr{\lnot\mathcal{E}_{j_0}} \leq \exp\inp{-c'k2^{j_0 -1}}$ for all
  large enough constant $k$ (to show this, one chooses the vector space
  $V$ in the statement of \Cref{cor:inf-B-perturb} to be
  $\ran{\mcblock{B}{j-1}}$).  Now, let $C > 0$ be as in
  \Cref{lem:max-norm}.  We choose $c_0$ to be small enough so that
  there exists $\gamma \in \inp{0, 1}$ satisfying
  \begin{equation}
    \label{eq:23}
    \kappa = \frac{\gamma}{\sqrt{1 + (c_0/C)^2}}.
  \end{equation}
  The claim of the theorem then follows if there exist positive
  constants $c_1, c_2$ such that for $k \geq c_1$,
  $\Pr{\lnot\mathcal{E}_\tau} \leq \exp\inp{-c_2k2^{j_0 - 1}}$.  We
  have already established this above for $j = j_0$.  We will show now
  that there exists a constant $c_2 > 0$ such that for large enough
  constant $k$ and $j \geq 2$,
  \begin{equation}
    \label{eq:21}
    \Pr{\lnot\mathcal{E}_j \vert \mathcal{E}_{j-1}}
    \leq \exp\inp{-c_2 k\cdot 2^{j-1}}.
  \end{equation}
  This will establish the claim if $c_1$ is chosen large enough that
  $\exp\inp{- c_2c_1} \leq \frac{1}{2}$, since in that case $k \geq
  c_1$ implies
  \begin{flexeq*}
    \Pr{\lnot{\mathcal{E}_\tau}} \leq
    \Pr{\lnot{\mathcal{E}_{j_0}}}
    + \sum_{j = j_0 + 1}^\tau \Pr{\lnot\mathcal{E}_j \vert \mathcal{E}_{j-1}}
    \eqbreak \leq \sum_{j=j_0}^\tau\exp\inp{-c_2 k\cdot 2^{j-1}}
    \leq 2 \exp\inp{- c_2 k 2^{j_0 - 1}}.
  \end{flexeq*}
  We now establish \cref{eq:21}.  Fix $j \geq j_0 + 1$, and assume
  $\mathcal{E}_{j-1}$ occurs. Note that
  \[
    \mcblock{B}{j}\vec{x} =
    \mblock{B}{j}\vblock{\vec{x}}{j}
    + \mcblock{B}{j - 1}\cblock{\vec{x}}{j-1},
  \]
  so that we can apply \Cref{lem:max-norm} with
  $V = \ran{\mcblock{B}{j-1}}$ and $\gamma$ as chosen above to find
  $c_1, c_2 > 0$ (not depending upon $j$) such that when
  $k \geq c_1$, it holds with probability at least
  $1 - \exp\inp{-c_2k\cdot 2^{j-1}}$ that
  \begin{flexeq}
    \label{eq:22}
    \frac{
      \norm[1]{\lclock{\mcblock{B}{j}\vec{x}}{j}}^2
    }{
      \norm[2]{\cblock{\vec{x}}{j, j_0}}^2
    }
    \geq \frac{1}{\tau^{8}}
    \max\Bigg\{
      C^2 \frac{\norm[2]{\vblock{\vec{x}}{j}}^2
      }{\norm[2]{\cblock{\vec{x}}{j, j_0}}^2},
    \eqbreak \frac{c_0^2\gamma^2\kappa^{2(j-j_0)}}{\kappa^2}
      \frac{
        \norm[2]{\cblock{\vec{x}}{j - 1, j_0}}^2
      }{\norm[2]{\cblock{\vec{x}}{j, j_0}}^2}
    \Bigg\}
    \eqbreak \quad \forall \vec{x} \neq \vec{0} \in \R^{2^j - 1}.
  \end{flexeq}
  Since
  \[
    \min_{0 \leq \eta \leq 1}\max\inb{a\eta, b(1-\eta)}
    = \frac{ab}{a + b},
    \text{ for all $a, b > 0$,}
  \]
  the guarantee in \cref{eq:22} implies that for all $\vec{x} \neq
  \vec{0}$ in $\R^{2^j - 1}$.
  \begin{flexeq*}
    \frac{
      \norm[1]{\lclock{\mcblock{B}{j}\vec{x}}{j}}
    }{
      \norm[2]{\cblock{\vec{x}}{j, j_0}}
    }
    \geq \frac{c_0\kappa^{j - 1 - j_0}}{\tau^4}
    \cdot\frac{\gamma}{\sqrt{1 + (\gamma c_0
        \kappa^{j- j_0 - 1}/C)^2}} \eqbreak\geq \frac{c_0\kappa^{j-j_0}}{\tau^4},
  \end{flexeq*}
  where the last inequality uses \cref{eq:23} and the fact that
  $\gamma, \kappa \leq 1$.  We thus have
  $\Pr{\mathcal{E}_j\vert \mathcal{E}_{j-1}} \geq 1 -
  \exp\inp{-c_2k\cdot 2^{j-1}}$, as required.
\end{proof}

We now use the information about $B$ derived above to show that $M$
comes very close to satisfying the conditions asked of an encoding
matrix in \Cref{thm:encoding-matrix}.  In particular,
\Cref{cor:invert-M} implies that $M$ satisfies these constraints at
any given fixed time $T$.  \Cref{cor:invert-all-sub-matrices-M} then
shows that encoding using $M$ actually satisfies, at \emph{each} time
$t$ up to the total time $T$ for which communication lasts, a slightly
weaker set of conditions which allow for the decoding of all but a
$\poly{\log t}$ sized suffix of the signal. Finally, we obtain the full statement of
 \Cref{thm:encoding-matrix} in \Cref{sec:code-matrix} by slightly modifying $M$ to handle the suffix differently.

\begin{corollary}[\textbf{Invertibility of $M$}]
  \label{cor:invert-M}
  For any $\delta \in (0, \frac{1}{2})$, there exist constants
  $c_0, c_1, c_2$ such that the following is true.  Let $T$ be a fixed
  integer, and let $\tau = \tf(T)$.  Let $k \geq c_1$ be a rate
  parameter.  Then, for $M$ sampled according to $\mathcal{A}_{T, k}$,
  we have
  \begin{flexeq*}
    \mathbb{P}\left[
      \exists \vec{x} \in \R^{T} \text{ s.t. }
      \norm[\dagger_\mu(kT)]{M\vec{x}} \right.
    \eqbreak < \left.
      c_0 2^{\delta j_0}2^{\tau(1/2 - \delta)}
      \norm[\star_\mu(T)]{\cblock{\vec{x}}{\tau, j_0}}/\tau^4
    \right]
    \eqbreak \leq \exp\inp{-c_2k\cdot 2^{j_0-1}}
  \end{flexeq*}
  for all $1 \leq j_0 \leq \tau$.  Here, for the purposes of computing
  the $\norm[\star_\mu(T)]{\cdot}$-norm, $\cblock{\vec{x}}{\tau, j_0}$ is seen
  as a vector in $\R^T$ whose last $2^{j_0 - 1} - 1$ coordinates are
  $0$.
\end{corollary}
\begin{proof}
  Using the same calculation as in \cref{eq:24}, we see that if $M \sim \mathcal{A}_{T, k}$, and
  $B$ is constructed from $M$ as defined in \cref{eq:def-B-from-M},
  then $B \sim \mathcal{B}_{T, k}$ and
  \begin{equation}
    \label{eq:33}
    \inf \frac{\norm[\dagger_\mu(kT)]{M\vec{x}}}{\norm[\star_\mu(T)]{\cblock{\vec{x}}{\tau, j_0}}}
    \geq 2^{\tau/2 - 1} \inf
    \frac{\norm[1]{B\vec{y}}}{\norm[2]{\cblock{\vec{y}}{\tau, j_0}}}.
  \end{equation}
  Given $\delta \in \inp{0, \frac{1}{2}}$, we choose
  $\kappa = 2^{-\delta}$.  After applying \Cref{thm:main-lower-b} with
  this value of $\kappa$ and using \cref{eq:33}, we then find positive
  constants $c_0, c_1, c_2$ (depending upon $\kappa$) such that when
  $k \geq c_1$, the matrix $M$ satisfies
  \begin{flexeq*}
    \norm[\dagger_\mu(kT)]{M\vec{x}}
    \geq c_0 2^{\delta j_0 + \tau(1/2 - \delta)}
    \norm[\star_\mu(T)]{\cblock{\vec{x}}{t, j_0}}/\tau^4\; \forall \vec{x}
  \end{flexeq*}
  with probability at least $1 - \exp\inp{-c_2k\cdot 2^{j_0 - 1}}$.
\end{proof}

\begin{corollary}[\textbf{Invertibility of principal submatrices of
    $M$}]
  \label{cor:invert-all-sub-matrices-M}
  For any $\delta \in (0, \frac{1}{2})$, there exist constants
  $c_0, c_1$ such that the following is true.  Let $T$ be a positive
  integer, and set $\tau \defeq \tf(T)$,
  $j_0(n) \defeq \ceil{\frac{4 \tf(\tf(n))}{\delta}} = \Theta(\log
  \log n)$.  Then, for any rate parameter $k \geq c_1$, there exists a
  \lot{k}\ matrix $M$ satisfying the following conditions. (Here, for
  $1 \leq n \leq T$ and a \lot{k}\ matrix $A$, $A_n$ denotes the
  \lot{k}\ matrix obtained by taking the first $n$ columns of $A$ and
  the first $kn$ rows).
  \begin{enumerate}
  \item \label{item:M-operator-norm} \textbf{\textup{Submatrices of
        $M$ have small operator norm}}:
    $\norm[2 \rightarrow 2]{M_n} \leq 1$ for $1 \leq n \leq T$.
  \item \label{item:M-invert} \textbf{\textup{Submatrices of $M$ are
        robustly invertible with respect to the past}}: for
    $1 \leq n \leq T$,
    \begin{displaymath}
      \norm[\dagger_\mu(kn)]{M_n\vec{x}} \geq {c_0 n^{(1/2 - \delta)}}
      \norm[\star_\mu(n)]{\cblock{\vec{x}}{\tf(n), j_0(n)}}.
    \end{displaymath}
    for all $\vec{x} \in \R^n$.
  \end{enumerate}
\end{corollary}
\begin{proof}
  Let $M \sim \mathcal{A}_{T, k}$.  We will show that when
  $k \geq c_1$ for $c_1$ large enough, then $M$ satisfies both the
  above conditions with positive probability.  We start by noting that
  \Cref{lem:op-norm-of-M} implies that \cref{item:M-operator-norm} is
  satisfied with probability at least $\frac{1}{2}$, as long as $c_1$
  is large enough.  We now turn to \cref{item:M-invert}.

  Each sub-matrix $M_n$ of $M$ is a sample from $\mathcal{A}_{n,k}$.
  From \Cref{cor:invert-M}, we therefore find constants $c_0, c_1'$
  and $c_2'$ such that as long as $k \geq c_1'$, we have
  \begin{flexeq}
    \label{eq:35}
    \mathbb{P}\left[\exists \vec{x} \in \R^n,
      \norm[\dagger_\mu(kn)]{M_n\vec{x}} \right.
    \eqbreak < \left.
      \frac{c_02^{\delta j_0(n)}
        2^{\tf(n)(1/2 - \delta)}}{\tf(n)^4}
      \norm[\star_\mu(n)]{\cblock{\vec{x}}{\tf(n), j_0(n)}}
    \right]
    \eqbreak \leq e^{-c_2'k 2^{j_0(n)}}
    \leq e^{-c_2'k\tf(n)^{8}},
  \end{flexeq}
  where the last inequality uses the value of $j_0$. (Note that,
  strictly speaking, we can only apply \Cref{cor:invert-M} when
  $j_0(n) \leq \tf(n)$.  However, when $j _0(n) > \tf(n)$, \cref{eq:35}
  is vacuously satisfied since in that case,
  $\cblock{\vec{x}}{\tf(n), j_0(n)}$ is an empty vector for all
  $\vec{x} \in \R^n$.)  Now, when $k \geq c_1$ where $c_1$ is chosen
  to be large enough that $c_2'c_1 \geq 10$ and $c_1 \geq c_1'$, we
  can substitute the value of $j_0(n)$ in \cref{eq:35} to find that
  for all $1 \leq n \leq T$
  \begin{flexeq*}
    \mathbb{P}\left[
      \exists \vec{x} \in \R^n, \norm[\dagger_\mu(kn)]{M_n\vec{x}}
    \right.
    \eqbreak < \left.
    {c_0 2^{\tf(n)(1/2 - \delta)}}
    \norm[\star_\mu(n)]{\cblock{\vec{x}}{\tf(n), j_0(n)}}
    \right]
    \eqbreak \leq e^{-10\tf(n)^{8}}
    \leq \frac{1}{10n^2}.
  \end{flexeq*}
  Taking a union bound over $1 \leq n \leq T$ and using
  $\sum_{n \geq 1}(1/n^2) < 2$, we now see that $M$ satisfies both
  conditions with probability at least
  $\frac{1}{2} - \frac{1}{5} = \frac{3}{10}$.
\end{proof}

\section{The encoding matrix}
\label{sec:code-matrix}
\Cref{cor:invert-all-sub-matrices-M} already contains most of the
information necessary for the construction of our encoding matrix.
Indeed, the matrix $M$ guaranteed there can already decode all but the
last $\poly{\log t}$ entries at any time $t$ with the required
guarantee.  To get the final guarantee, we only need to make our
encoding systematic by including a copy of the input symbols.  More
precisely, given a \lot{k} matrix of the form guaranteed by
\Cref{cor:invert-all-sub-matrices-M}, we construct a \lot{(k+1)}
matrix $C$ which at time $t$ produces the $k$ symbols that would have
been output by $M$, followed by the current input $x_t$.  In symbols,
this means that entries of $C$ can be written as follows (we use again
the block notation for row indices of \lot{k} matrices introduced in
\Cref{sec:online-codes-low}):
\begin{equation}
  \label{eq:59}
  C_{(i, l), j} = \begin{cases}
      M_{(i. l), j} & \text { when $1 \leq l \leq k$,}\\
      1 & \text{ when $l = k + 1$, and}\\
      0 & \text{otherwise}
    \end{cases}
\end{equation}
We note the following simple consequences of this definition:
\begin{enumerate}
\item Let $y = Mx$ for $x \in \R^T$.  Then
  \begin{align}
    \notbool{j2col}{}{&}
    \norm[\dagger_\mu((k+1)T)]{Cx}
    \notbool{j2col}{}{ \nonumber \\}
    &=
      \sum_{\substack{1\leq i \leq T\\1 \leq l \leq k}}
    \abs{y_{(i,l)}} \inp{\frac{(k+1)T}{(k+1)(T - i+1) - l +
    1}}^{(1-\mu)/2}\nonumber\\
    &\qquad\qquad + \sum_{i=1}^T
      \abs{x_i} \inp{\frac{(k+1)T}{(k+1)(T - i + 1) - k}}^{(1-\mu)/2} \nonumber\\
    & \geq
      \frac{\norm[\dagger_\mu(kT)]{y}}{3}
      + \sum_{i=1}^T \abs{x_i} \inp{\frac{T}{T - i + 1}}^{(1- \mu)/2} . \label{eq:60}
  \end{align}
\item For every $\vec{x} \in \R^T$,
  \begin{equation}
    \label{eq:37}
    \norm[2]{Cx}^2 = {\norm[2]{Mx}^2 + \norm[2]{x}^2}.
  \end{equation}
\end{enumerate}
We can now prove \Cref{thm:encoding-matrix} which we restate here for
easy reference.
\begin{theorem*}[\textbf{The encoding matrix, restatement of \Cref{thm:encoding-matrix}}]
  For any $\mu \in (0, 1]$ and $\delta \in (0, \frac{1}{2})$, there
  exist constants $c, c_1$ such that the following is true.  Let
  $T \geq 3$ be any integer.  For a rate parameter $k$ satisfying
  $k \geq c_1$, there exists a matrix $\mathcal{C}$ satisfying the
  following conditions. (Here, for $1 \leq n \leq T$ and a \lot{k}\
  matrix $A$, $A_n$ denotes the leading principal sub-matrix of $A$
  consisting of its first $n$ columns and $kn$ rows).
  \begin{enumerate}
  \item %
    \textbf{\textup{Submatrices of $\C$ have small operator norm}}:
    $\norm[2 \rightarrow 2]{\C_n} \leq 1$ for $1 \leq n \leq T$.
  \item %
    \textbf{\textup{Submatrices of $\C_n$ are robustly invertible}}:
    for $1 \leq n \leq T$,
    \begin{displaymath}
      \norm[\dagger_{\mu}(kn)]{\C_n\vec{x}} \geq c n^{(1/2 - \delta)}
      \norm[\star_\mu(n)]{\vec{x}} \text{ for all $\vec{x} \in \R^n$}.
    \end{displaymath}
  \end{enumerate}
\end{theorem*}

\begin{proof}
  Applying \Cref{cor:invert-all-sub-matrices-M} with $\mu$ and
  $\delta$, we obtain $c_0, c_1, j_0$ and a \lot{k} matrix $M$ (for
  a $k \geq c_1$) as in the corollary.  We define the \lot{k+1}
  matrix $C$ using $M$ as done in \cref{eq:59} above, and set
  $\mathcal{C} = \frac{1}{\sqrt{2}}C$.  \Cref{item:C-operator-norm}
  now follows from \cref{item:M-operator-norm} and \cref{eq:37}.

  For \cref{item:C-invert}, we use \cref{eq:60} followed by
  \cref{item:M-invert} of \Cref{cor:invert-all-sub-matrices-M} to get (with $c' \defeq c_0/3)$
  \begin{flexeq}
    \label{eq:61}
    \sqrt{2}\norm[\dagger_\mu(kn)]{\C_n\vec{x}}
    \eqbreak \geq c' n^{(1/2 - \delta)}
    \norm[\star_\mu(n)]{\cblock{\vec{x}}{\tf(n), j_0(n)}}
    \eqbreak + \sum_{i=1}^n \inp{\frac{n}{n - i
        + 1}}^{(1- \mu)/2} \abs{x_i},
    \eqbreak \text{ for all $\vec{x} \in \R^n$,}
  \end{flexeq}
  where $\tf(n) \defeq \ceil{\lg(n + 1)}$, and $j_0$ is as in
  \Cref{cor:invert-all-sub-matrices-M} and satisfies
  $j_0(n) = O(\log \log n)$.  We now estimate the second term as
  follows:
  \begin{align*}
    \notbool{j2col}{}{&}
    \sum_{i=1}^n \inp{\frac{n}{n - i + 1}}^{(1- \mu)/2} \abs{x_i}
    \notbool{j2col}{}{\\}
    &\geq \sum_{i=n - 2^{j_0(n)} + 2}^n\frac{\sqrt{n}}{\sqrt{n - i +
      1}}\cdot \inp{\frac{n-i+1}{n}}^{\mu/2} \abs{x_i}\\
    &\geq \frac{\sqrt{n}}{O(\poly{\log n})}
      \sqrt{\sum_{i=n - 2^{j_0(n)} + 2}^n\inp{\frac{n-i+1}{n}}^{\mu}
      \abs{x_i}^2}\\
    &\geq c''n^{(1/2 - \delta)} \norm[\star_\mu(n)]{\cblock{\vec{x}}{j_0(n)-1}},
  \end{align*}
  for some fixed positive constant $c'' = c''(\delta)$.
  \Cref{item:C-invert} now follows by substituting this into
  \cref{eq:61} and using the fact that
  \begin{multline*}
    \norm[\star_\mu(n)]{\cblock{\vec{x}}{j_0(n)-1}}
    + \norm[\star_\mu(n)]{\cblock{\vec{x}}{\tf(n), j_0(n)}}
    \\
    \geq \sqrt{
      \norm[\star_\mu(n)]{\cblock{\vec{x}}{j_0(n)-1}}^2
      + \norm[\star_\mu(n)]{\cblock{\vec{x}}{\tf(n), j_0(n)}}^2
    } \notbool{j2col}{}{\\}
    = \norm[\star_\mu(n)]{\vec{x}}.\qedhere
  \end{multline*}
\end{proof}

\section{Comparing the online and block settings: A lower bound}
\label{sec:lower-bound}
As noted in the introduction, when compared with the block coding
setting, we lose an extra $\Theta(t^\delta)$ factor in the robust
invertibility guarantee in the online setting.  A natural question
therefore is whether it is possible to get rid of this loss and obtain
a guarantee as strong as the block coding setting (\cref{eq:32}) in
the online setting as well.  We now prove \Cref{prop:lower-bound},
which was stated in the introduction as a partial answer to this
question.  We restate the theorem here for ease of reference.

\begin{theorem}
  \label{prop:lower-restate}
  Fix $\mu \in (0, 1]$, $c_0 > 0$ and a positive integer $k$.  There
  exists a constant $\tau = \tau(\mu, c_0, k)$ such that the following
  is true.  If $\C$ is a $kT \times T$ \lot{k} matrix such that for
  all $t \in [T]$ the submatrix $\C_t$ of $\C$ satisfies
  \begin{equation}
    \norm[\dagger_\mu(kt)]{\C_t\vec{x}} \geq {c_0 t^{1/2}}
    \norm[\star_{\mu}(t)]{\vec{x}} \text{ for all $\vec{x} \in \R^t$},\label{eq:62}
  \end{equation}
  then there exists a non-zero $\vec{x} \in \R^T$ for which
  \begin{displaymath}
    \norm[2]{\C\vec{x}}^2 \geq \tau \norm[2]{\vec{x}}^2 \sum_{i=1}^T\frac{1}{i}
    \geq \inp{\tau\log T}\cdot\norm[2]{\vec{x}}^2.
  \end{displaymath}
  In particular, $\vec{x}$ can be taken to be the unit pulse at time $1$.
\end{theorem}

\begin{proof}
  When $\vec{x} = e_1$ is the unit pulse at time $1$, we have
  $\norm[\star_\mu(t)]{\vec{x}_{[t]}} = 1$  for all $\mu \in (0, 1]$ and $1
  \leq t \leq T$.  Let $\vec{z} \in \R^{kT}$ be the vector such that $z_i
  \defeq \abs{\C_{i,1}}$.  Then, for $x = e_1$, we have
  \begin{align*}
    \norm[2]{\C\vec{x}} &= \norm[2]{\vec{z}}\text{, and }\\
    \norm[\dagger_\mu(kt)]{\C_tx_{[t]}}
    &= \sum_{i = 1}^{kt} z_i \cdot \inp{\frac{kt}{kt - i + 1}}^{(1-\mu)/2}.
  \end{align*}
  It therefore follows that when the guarantees of \cref{eq:62} are
  enforced, the objective value of the following convex program is a
  lower bound on $\norm[2]{\C\vec{e_1}}$:
  \begin{equation}
    \label{eq:63}
    \begin{aligned}
      &\min &\qquad \norm[2]{\vec{z}}^2\\
      &\text{subject to}
      &\qquad \sum_{i=1}^{kt} z_if_{it} & \geq \gamma_t\text{, }&1 \leq t \leq T\\
      &&\qquad z _i &\geq 0,\text{ }&1 \leq i \leq kT.
    \end{aligned}
  \end{equation}
  Here
  \[
    \gamma_t \defeq \frac{c_0t^{\mu/2}}{k^{(1-\mu)/2}}\text{, and }
    f_{it} \defeq \frac{1}{(kt - i + 1)^{(1-\mu)/2}}.
  \]
  We will lower bound the objective value of this program by providing
  a feasible solution to its dual program.\footnote{Note that since
    the primal objective function is convex in $\vec{z}$ and the since
    the primal constraints admit a feasible point where all
    constraints are satisfied with a strict inequality, Slater's
    constraint qualifications are satisfied.  Thus, strong duality
    also holds, though it is not required for our
    purposes.} The dual program is given as
  \begin{equation}
      \label{eq:64}
    \begin{aligned}
      &\sup &\qquad g(\vec{\lambda}, \vec{\nu})\\
      &\text{subject to}
      &\qquad \lambda_i &\geq 0 \text{, }&1 \leq i \leq T\\
      &&\qquad \nu_i &\geq 0,\text{ }&1 \leq i \leq kT
    \end{aligned}
  \end{equation}
  where
  \begin{displaymath}
    g(\vec{\lambda}, \vec{\nu}) \defeq \inf_{\vec{z} \in \R^{kT}}
    \norm[2]{\vec{z}}^2
    - \sum_{t=1}^T\lambda_t\inp{\sum_{i=1}^{kt}z_if_{it} - \gamma_t}
    -\sum_{i=1}^{kT}\nu_iz_i.
  \end{displaymath}
  The expression to be minimized in the definition of $g$ is a convex
  function of $\vec{z}$, and hence we can perform the minimization by
  equating the gradient to $0$.  This yields
  \begin{equation}
    \label{eq:65}
    g(\lambda, \nu) = \sum_{t=1}^T\lambda_t\gamma_t -
    \frac{1}{4}\sum_{t=1}^{kT}\Lambda_i^2 -
    \frac{1}{4}\sum_{i=1}^{kT}\nu_i^2 -
    \frac{1}{2}\sum_{t=1}^{kT}\Lambda_t \nu_t,
  \end{equation}
  where, for $1 \leq i \leq kT$,
  \[
    \Lambda_i \defeq \sum_{t = \ceil{i/k}}^T\lambda_tf_{it}
    = \sum_{t = \ceil{i/k}}^T
    \frac{\lambda_t}{(kt -i +1)^{(1-\mu)/2}}.
  \]
  Note that when
  $\vec{\lambda}$ and $\vec{\nu}$ are non-negative, $g$ is
  non-increasing in the $\nu_i$, and hence we can set $\nu_i = 0$ (for
  $1 \leq i \leq kT$) without changing the optimal value of the program
  in \eqref{eq:64}.  We now consider the following dual feasible solution:
  \begin{equation}
    \begin{aligned}
      \lambda_t
      &= \frac{a_0}{t^{1 + \mu/2}} &&\text{for $1 \leq t \leq T$, and}\\
      \nu_i &= 0 && \text{for $1 \leq i \leq kT$,}
    \end{aligned}\label{eq:67}
  \end{equation}
  where $a_0$ is a positive constant to be chosen later.  To lower
  bound the dual objective value, we now upper bound the $\Lambda_i$
  given this choice of the $\lambda_t$.  For positive integers $i$ and
  $j$ such that $1 \leq j \leq T$ and $k(j-1) + 1 \leq i \leq kj$, we
  have
  \begin{align}
    \Lambda_i
    &= a_0\sum_{t = j}^T\frac{1}{t^{1 + \mu/2}(kt - i +
      1)^{(1-\mu)/2}} \nonumber\\
    & \leq a_0
      \sum_{t=j}^T\frac{1}{t^{1+\mu/2}\cdot(kt - kj + 1)^{(1-\mu)/2}}\nonumber\\
    &\leq \frac{a_0}{j^{1+\mu/2}} +
      \frac{a_0}{k^{(1-\mu)/2}}\sum_{t=1}^{T -
      j}\frac{1}{(t+j)^{1+\mu/2}\cdot t^{(1-\mu)/2}}.
      \label{eq:68}
  \end{align}
  The last term above can also be shown to be $O(\sqrt{j})$, as
  follows:
  \begin{align*}
    \notbool{j2col}{}{&}
    \sum_{t=1}^{T - j}\frac{1}{(t+j)^{1+\mu/2}\cdot t^{(1-\mu)/2}}
    \notbool{j2col}{}{\\}
    &\leq\sum_{t=1}^{\infty}\frac{1}{(t+j)^{1+\mu/2}\cdot
      t^{(1-\mu)/2}} \\
    &  =\sum_{l=1}^\infty\sum_{t=(l-1)j+1}^{lj}\frac{1}{(t+j)^{1+\mu/2}\cdot
      t^{(1-\mu)/2}}\\
    &\leq\sum_{l=1}^{\infty}\frac{1}{(lj)^{1+\mu/2}}
      \sum_{t=(l-1)j+1}^{lj}\frac{1}{t^{(1-\mu)/2}}\\
    &\leq\frac{2}{1+\mu}
      \sum_{l=1}^{\infty}\frac{j^{(1+\mu)/2}}{(lj)^{1+\mu/2}}\cdot
      \inp{l^{(1+\mu)/2} - (l-1)^{(1+\mu)/2}} \\
    &\leq \frac{2}{1+\mu}\cdot\frac{1}{\sqrt{j}}
      \sum_{l=1}^\infty \frac{1}{l^{1+\mu/2}}
    \leq \frac{2(2+\mu)}{\mu(1+\mu)}\cdot\frac{1}{\sqrt{j}}.
  \end{align*}
  Here, the third and the last inequalities use \Cref{fct:estimate}
  (note that $\mu > 0$, so only the case $\alpha \neq 1$ of
  \Cref{fct:estimate} is used), and the fourth uses the fact that when $\beta \in (0, 1]$ and
  $n$ is a non-negative integer, $(n+1)^{\beta} - n^\beta \leq 1$.
  Plugging the above estimate into \cref{eq:68}, we get that when $j$
  is a positive integer such that $k(j-1) + 1 \leq i \leq kj$,
  \begin{displaymath}
    \Lambda_i \leq \frac{a_0c'}{\sqrt{j}},
  \end{displaymath}
  where
  $c' = c'(\mu,k) \defeq 1 + \frac{2(2+\mu)}{\mu(1+\mu)\cdot
    k^{(1-\mu)/2}}$.  Thus, at the feasible solution in \cref{eq:67},
  the dual objective value is
  \begin{align*}
    g(\lambda, \mathbf{0})
    &= \sum_{t=1}^T\lambda_t\gamma_t -
      \frac{1}{4}\sum_{i=1}^{kT}\Lambda_i^2 \\
    & = \frac{a_0c_0}{k^{(1-\mu)/2}}\sum_{t=1}^{T}\frac{1}{t} -
      \frac{1}{4}\sum_{i=1}^{kT}\Lambda_i^2\\
    &\geq a_0\inp{\frac{c_0}{k^{(1-\mu)/2}}
      - \frac{ka_0{c'}^2}{4}}\sum_{t=1}^{T}\frac{1}{t},
  \end{align*}
  where the last inequality uses the above estimate on $\Lambda_i$.
  Thus, by choosing $a_0 = a_0(\mu, c_0, k)$ to be
  $\frac{2c_0}{k^{(3-\mu)/2}c'^2}$, we find that there exists a
  positive constant $\tau = \tau(\mu, c_0, k)$ such that the dual
  objective value is at least $\tau \sum_{t=1}^T\frac{1}{t}$. By weak
  duality, this is also a lower bound on the objective value of the
  primal program in \eqref{eq:63}. By the discussion preceding
  \eqref{eq:63}, this completes the proof.
\end{proof}

The proof of \Cref{prop:lower-restate} uses the following elementary estimate.
\begin{fact}
  \label{fct:estimate}
  Let $a < b$ be positive integers and $\alpha$ a positive real
  number.  Then,
  \begin{displaymath}
    \sum_{i = a + 1}^b i^{-\alpha} \leq \int_{a}^b x^{-\alpha} dx =
    \begin{cases}
      \frac{b^{1 - \alpha} - a^{1 - \alpha}}{1 - \alpha} &
      \textup{when $\alpha \neq 1$,}\\
      \log (b/a) & \textup{when $\alpha = 1$.}
    \end{cases}
  \end{displaymath}
\end{fact}

\appendix

\section{$\norm[1]{\cdot}$-norms of non-uniform Gaussian vectors}
\label{sec:norm1cdot-norms-non}
In this section we collect concentration bounds for the
$\norm[1]{\cdot}$-norms of Gaussian vectors with independent but not
identically distributed entries.  The bounds here are adaptations of
standard arguments and results in the literature on Gaussian
concentration to our setting.

We begin with the following elementary fact and a consequence, and
then proceed to bounds for the lower tail of the $\norm[1]{\cdot}$
norm of Gaussian vectors with independent but not identically
distributed entries (in \Cref{cor:conc,thm:l1-concentration}).
\begin{fact}[\textbf{Gaussian tail}]
  If $X \sim \mathcal{N}(0, \sigma^2)$, then for $t > 0$,
  \begin{flexeq*}
    \Pr{X \geq t}
    = \frac{1}{\sqrt{2\pi\sigma^2}}
    \int\limits_t^\infty \exp\inp{-x^2/(2\sigma^2)}dx %
    \eqbreak < \frac{\sigma}{t\cdot\sqrt{2\pi}}
    \exp\inp{-\frac{t^2}{2\sigma^2}}.
  \end{flexeq*}
  \label{lem:gaussian-tail}
\end{fact}

\begin{corollary}[\textbf{Upper tail of the $\norm[1]{\cdot}$-norm}]
  \label{lem:gauss-l1-upper}
  Let $X \sim \mathcal{N}\inp{0, \mathop{\mathrm{diag}}\inp{(\sigma_i^2)_{i=1}^n}}$.
  Then, for any $t > 0$ and $c > 0$, we have
  \begin{displaymath}
    \Pr{\norm[1]{X} > t} \leq \exp\inp{-ct + n \log 2 + \frac{c^2}{2}\sum_{i=1}^n\sigma_i^2)}.
  \end{displaymath}
  In particular, choosing $c = \frac{t}{\sum_{i=1}^n \sigma_i^2}$ and
  then $t = \alpha \sqrt{n}$, we have
  \begin{displaymath}
    \Pr{\norm[1]{X} > \alpha\sqrt{n}} \leq
    \exp\inp{-n\inp{\frac{\alpha^2}{2\sum_{i=1}^n\sigma_i^2} - \log 2}}.
  \end{displaymath}
\end{corollary}

\begin{proof}
  For $Y \sim \mathcal{N}(0, \sigma^2)$, we have, for any $c > 0$,
  $\E{\exp(c\abs{Y})} \leq 2\exp(c^2\sigma^2/2)$.
  Thus we have, for any $c > 0$,
  \begin{flexeq*}
    \Pr{\norm[1]{X} > t} = \Pr{\exp(c\norm[1]{X}) > \exp(ct)}
    \eqbreak \leq \exp(-ct)\prod_{i=1}^n\E{\exp(c\abs{X_i})}
    \eqbreak \leq \exp\inp{-ct + n\log 2 +
      \frac{c^2}{2}\sum_{i=1}^n\sigma_i^2}. \qedhere
  \end{flexeq*}
\end{proof}

\subsection{The lower tail of the $\norm[1]{\cdot}$-norm}
\label{sec:lower-tail-norm1cdot}
We now state two concentration results for the lower tail of the
$\ell_1$ norm of Gaussian vectors with independent but not identically
distributed entries.  The first (\Cref{cor:conc}) deals with the lower
tail for mean $0$ vectors (in other words, this is an upper
  bound on small-ball probability), while the second
(\Cref{thm:l1-concentration}) considers the concentration around the
$\ell_1$ norm of the mean for vectors with non-zero mean.

\begin{lemma}
  \label{lem:gauss-l1}
  Let $X \sim \mathcal{N}\inp{0, \mathop{\mathrm{diag}}\inp{(\sigma_i^2)_{i=1}^n}}$.
  Then, for any $t > 0$ and $c > 0$, we have
  \begin{displaymath}
    \Pr{\norm[1]{X} < t} \leq \exp(ct + \sum_{i=1}^n\nu(c^2\sigma_i^2)),
  \end{displaymath}
  where for $x \geq 0$,
  \begin{flexeq*}
   \nu(x) \defeq \frac{x}{2} - \frac{1}{2}\log \frac{\pi}{2} + \log
   \int\limits_{\sqrt{x}}^{\infty}\exp(-t^2/2)dt
   \eqbreak \leq \frac{1}{2}\min\inb{0,
      -\log \frac{\pi x}{2}}.
  \end{flexeq*}
\end{lemma}

\begin{proof}
  Let $X = (X_1, X_2, \dots, X_n)$ where $X_i \sim \mathcal{N}(0,
  \sigma_i^2)$.  For any $c > 0$, we have
  \begin{flexeq}
    \label{eq:6}
    \Pr{\norm[1]{X} < t} = \Pr{\exp\inp{-c\norm[1]{X}} > \exp(-ct)}
    \eqbreak \leq
    \exp(ct) \cdot \prod_{i=1}^n\E{\exp\inp{-c\abs{X_i}}}.
  \end{flexeq}
  The first claim now follows since for
  $Y \sim \mathcal{N}(0, \sigma^2)$, we have
  \begin{align*}
    \E{\exp\inp{-c\abs{Y}}} &=
    \sqrt{\frac{2}{\pi\sigma^2}}
    \int\limits_0^\infty\exp\inp{-cy - \frac{y^2}{2\sigma^2}}dy \\
    &= \sqrt{\frac{2}{\pi}}\exp(c^2\sigma^2/2) \notbool{j2col}{}{\\}
    \notbool{j2col}{}
    {&\qquad\cdot}
       \int\limits_0^\infty\exp\inp{-(z + c\sigma)^2/2}dz\\
    &= \sqrt{\frac{2}{\pi}}\exp(c^2\sigma^2/2)
      \int\limits_{c\sigma}^\infty\exp\inp{-z^2/2}dz \\
                            &= \exp\inp{\nu(c^2\sigma^2)}.
  \end{align*}
  The definition of $\nu$ implies that $\nu(x) \leq 0$ for all positive
  $x$.  Now, using \Cref{lem:gaussian-tail}, we have
  \begin{flexeq*}
    \nu(x) \leq -\frac{1}{2}\log\inp{\frac{\pi}{2}} + \frac{x}{2} +
    \log\inp{\frac{1}{\sqrt{x}}\exp\inp{-x/2}}
    \eqbreak = -\frac{1}{2}\log
    \inp{\frac{\pi x}{2}}\text{, for all $x \geq 0$.}
  \end{flexeq*}
  Thus, we obtain
  $\nu(x) \leq \frac{1}{2}\min\inb{0, -\log\frac{\pi x}{2}}$.
\end{proof}

\begin{theorem}[\textbf{Lower tail of the $\norm[1]{\cdot}$-norm}]
  \label{cor:conc}
  There exists a positive constant $\gamma$ such that the following is
  true. Let
  $X \sim \mathcal{N}\inp{0,
    \mathop{\mathrm{diag}}\inp{(\sigma_i^2)_{i=1}^n}}$, $S$ an
  arbitrary subset of $[n]$.  Define
  $G \defeq GM\inp{\inp{\sigma_i}_{i\in S}}$ to be the geometric mean
  of the $\sigma_i$ for $i \in S$. Then, for all $\tau \geq 0$ and
  \[t \leq \tau \gamma G |S|\] we have
  \begin{displaymath}
    \Pr{\norm[1]{X} < t} \leq \inp{\tau}^{|S|}.
  \end{displaymath}
  In particular, given $\alpha$, if there exists a set $S$ for which
  $G \geq \alpha/|S|$, then for all $\tau \geq 0$,
  \begin{displaymath}
    \Pr{\norm[1]{X} < \tau\gamma\cdot \alpha} \leq \tau^{|S|}.
  \end{displaymath}
\end{theorem}
\begin{proof}
  Set $\gamma = \sqrt{\pi/2}/e$.  We now use \Cref{lem:gauss-l1} and
  the upper bound on the function $\nu$ defined there, after exercising
  our choice for $c$ by setting $c = |S|/t$.  We then have
  \begin{align*}
    \Pr{\norm[1]{X} < t}
    &\leq \exp\big(ct + |S|\log(\sqrt{2/\pi}) \notbool{j2col}{}{\\}
    \notbool{j2col}{}{
    &\qquad\qquad\qquad
      }
      - |S|\log c - |S| \log G\big)\\
    &= \exp\inp{|S|\inp{-\log \gamma - \log |S| + \log t -
      \log |G|}} \\
    &= \inp{\frac{t}{\gamma G |S|}}^{|S|}.
  \end{align*}
  Substituting $t = \tau \gamma |G| S$, we get the claimed result.
\end{proof}

The standard fact below shows that it is sufficient to consider mean
$0$ vectors in the setting of \Cref{cor:conc}.  We include a proof for
completeness.
\begin{fact}[\textbf{Stochastic domination of absolute values of
    Gaussians}]
  Let $X \sim \mathcal{N}(0, \sigma^2)$.  Then the random variable
  $\abs{a + X}$ stochastically dominates the random variable $\abs{b +
    X}$ whenever $\abs{a} > \abs{b}$.
  \label{fct:stochastic}
\end{fact}
\begin{proof}
  Without loss of generality, we assume $\sigma = 1$ and $a > b > 0$.
  Now for any fixed $y \geq 0$ and $x \in [b, a]$, we have
  \begin{displaymath}
    \Pr{\abs{x + X} \geq y} = f(x) \defeq G(y + x) + G(y - x),
  \end{displaymath}
  where
  $G(x) \defeq
  \frac{1}{\sqrt{2\pi}}\int\limits_{x}^{\infty}\exp\inp{-t^2/2}dt$ is
  the Gaussian tail.  The claim now follows from the calculation that
  \begin{displaymath}
    f'(x) = \sqrt{\frac{2}{\pi}}\exp\inp{-(x^2 + y^2)/2}\sinh(xy) \geq 0
  \end{displaymath}
  for $x, y \geq 0$.
\end{proof}
A standard coupling argument gives the following corollary.
\begin{corollary}
  \label{cor:stochastic-domination}
  Let
  $\vec{X} \sim \mathcal{N}\inp{0,
    \mathop{\mathrm{diag}}\inp{(\sigma_i^2)_{i=1}^n}}$, and
  $\vec{v} \in \R^{n}$.  For any $t \geq 0$, we have
  \begin{displaymath}
    \Pr{\norm[1]{\vec{v} + \vec{X}} < t} \leq \Pr{\norm[1]{\vec{X}} < t}.
  \end{displaymath}
\end{corollary}

\begin{theorem}[\textbf{Lower tail of the $\norm[1]{\cdot}$-norm with
    non-zero means}]
  \label{thm:l1-concentration}
  For any $\gamma \in (0, 1)$, there exists a positive constant
  $c = c(\gamma)$ such that the following is true.  Let
  $\vec{X} \sim \mathcal{N}(0, \inp{\sigma_i}_{i=1}^n)$ be a Gaussian
  random vector with mean $0$ and independent co-ordinates with
  non-zero variance, and let $\vec{a} \in \R^n$ be an arbitrary
  vector.
  \begin{displaymath}
    \Pr{\norm[1]{\vec{X} + \vec{a}}
      < \gamma \norm[1]{\vec{a}}}
    \leq \exp\inp{-c \frac{(\sum_i \sigma_i)^2}{\sum_i \sigma_i^2}}.
  \end{displaymath}
\end{theorem}

\begin{proof}
  A direct calculation (or the fact that the map
  $X \mapsto \abs{X + a}$ is 1-Lipschitz and the
  Cirel'son-Ibragimov-Sudakov inequality~(\cite{76:_cirel_b}, as
  stated in \cite[Theorem 3.2.2]{raginsky_concentration_2013}) implies
  that each $\abs{X_i + a_i}$ is a sub-gaussian random variable with
  mean $\mu_i = \E{\abs{X_i + a_i}}$ and sub-gaussian parameter
  $\sigma_i$.  Further, note that $\mu_i \geq \abs{a_i}$ (due to
  Jensen's inequality) and
  $\mu_i \geq \E{\abs{X_i}} = \sigma_i\sqrt{2/\pi}$ (by
  \Cref{fct:stochastic}).

  Since the $X_i$ are independent, this implies that
  $Z \defeq \norm[1]{X + a} = \sum_i \abs{X_i + a_i}$ is also a
  sub-gaussian random variable with mean $\sum_i{\mu_i}$ and
  sub-Gaussian parameter $\sqrt{\sum_i \sigma_i^2}$.  Further, since
  $\mu_i \geq \abs{a_i}$, we have $\E{Z} \geq \norm[1]{a}$, so that
  \begin{flexeq*}
    \Pr{Z \leq \gamma \norm[1]{a}} \leq \Pr{Z \leq \gamma E[Z]}
    \eqbreak \leq \exp\inp{-c'(1-\gamma)^2 \E{Z}^2/\inp{\sum_i
        \sigma_i^2}},
  \end{flexeq*}
  where the second inequality uses the fact that $Z$ is sub-gaussian
  with sub-gaussian parameter $\sqrt{\sum_i \sigma_i^2}$.  The claim
  now follows once we recall that $\mu_i \geq \sigma_i\sqrt{2/\pi}$ so
  that $\E{Z} \geq \sqrt{2/\pi}\sum_i \sigma_i$.
\end{proof}

\section*{Acknowledgments}
We thank anonymous reviewers for several helpful comments and suggestions.

\bibliographystyle{alphaurl}
\bibliography{encodings}

\end{document}